\documentclass{fundam}
\usepackage{amsfonts}

\usepackage{tabularx,lipsum,environ}

\usepackage{url} 
\usepackage[ruled,lined]{algorithm2e}
\usepackage{graphicx}
\usepackage{tikz}
\usetikzlibrary{automata, positioning, arrows}
\usetikzlibrary{graphs,quotes,petri, shapes}
\makeatletter
\newcommand{\oproblemtitle}[1]{\gdef\@oproblemtitle{#1}}
\newcommand{\oprobleminput}[1]{\gdef\@oprobleminput{#1}}
\newcommand{\oproblemsolution}[1]{\gdef\@oproblemsolution{#1}}
\newcommand{\oproblemmeasure}[1]{\gdef\@oproblemmeasure{#1}}
\NewEnviron{optproblem}{
  \oproblemtitle{}\oprobleminput{}\oproblemsolution{}\oproblemmeasure{}
  \BODY
  \par\addvspace{.5\baselineskip}
  \noindent
  \begin{tabularx}{0.97\textwidth}{@{\hspace{\parindent}} l X c}
    \multicolumn{3}{@{\hspace{\parindent}}l}{\@oproblemtitle} \\
    \emph{Input:} & \@oprobleminput \\
    \emph{Solution:} & \@oproblemsolution\\ 
    \emph{Measure:} & \@oproblemmeasure\\
  \end{tabularx}
  \par\addvspace{.5\baselineskip}
}
\makeatother

\makeatletter
\newcommand{\drop}[1]{}
\newcommand{\problemtitle}[1]{\gdef\@problemtitle{#1}}
\newcommand{\probleminput}[1]{\gdef\@probleminput{#1}}
\newcommand{\problemquestion}[1]{\gdef\@problemquestion{#1}}
\NewEnviron{decisionproblem}{
  \problemtitle{}\probleminput{}\problemquestion{}
  \BODY
  \par\addvspace{.5\baselineskip}
  \noindent
  \begin{tabularx}{0.97\textwidth}{@{\hspace{\parindent}} l X c}
    \multicolumn{2}{@{\hspace{\parindent}}l}{\@problemtitle} \\
    \emph{Input:} & \@probleminput \\
    \emph{Question:} & \@problemquestion
  \end{tabularx}
  \par\addvspace{.5\baselineskip}
}
\makeatother

\newlength{\problemoffset}
\newcommand{\T}{\ensuremath{\mathcal{T}}}

\newcommand{\escale}[1]{\ensuremath{\textbf{\scalebox{0.8}{#1}}}}
\newcommand{\nscale}[1]{\ensuremath{\textbf{\scalebox{0.8}{#1}}}}

\newcommand{\myEdge}[2]{ \tikz[baseline=-3pt]{
\draw[#2,line width=0.5pt] (0,0) -- ++(0.6,0) node[anchor=base, yshift=3pt, pos=0.5] {\escale{$#1$}};
}}

\newcommand{\Edge}[1]{ \tikz[baseline=-1pt]{
\draw[->,line width=0.5pt] (0,0) -- ++(0.6,0) node[anchor=base, yshift=4pt, pos=0.5] {\escale{$#1$}};
}}

\newcommand{\LTS}{LTS}
\newcommand{\edge}[1]{\myEdge{#1}{->}}
\newcommand{\nedge}[1]{\edge{\neg{#1}}}

\newcommand{\U}{\ensuremath{\mathfrak{U}}}
\newcommand{\R}{\ensuremath{\mathfrak{R}}}
\newcommand{\E}{\ensuremath{\mathfrak{E}}}
\newcommand{\K}{\ensuremath{\mathfrak{K}}}
\newcommand{\mS}{\ensuremath{\mathfrak{S}}}
\newcommand{\N}{\ensuremath{\mathbb{N}}}
\newcommand{\eff}{\text{eff}}
\newcommand{\w}{\text{W}}
\newcommand{\opt}{\ensuremath{\texttt{opt}}}
\newcommand{\minnpo}{\textsc{Min~NPO}}
\tikzstyle{place}=[circle,thick,draw=blue!75,fill=blue!20,minimum size=6mm]
  \tikzstyle{red place}=[place,draw=red!75,fill=red!20]
  \tikzstyle{transition}=[rectangle,thick,draw=black!75,
  			  fill=black!20,minimum size=4mm]
			   \tikzstyle{every label}=[red]

\begin{document}

\setcounter{page}{167}
\publyear{22}
\papernumber{2136}
\volume{187}
\issue{2-4}

 \finalVersionForARXIV

\title{Some Basic Techniques Allowing Petri Net Synthesis: \\
                     Complexity and Algorithmic Issues}

\author{Raymond Devillers\\
D\'epartement d'Informatique - Universit\'e Libre de Bruxelles\\
 Boulevard du Triomphe,  B1050 Brussels, Belgium\\
 rdevil@ulb.ac.be
 \and Ronny Tredup\thanks{Address of correspondence:  Institut F\"ur Informatik - Universit\"at Rostock, Albert-Einstein-Stra{\ss}e 22,
                              D18059 Rostock, Germany}
 \\
 Institut F\"ur Informatik - Universit\"at Rostock\\
  Albert-Einstein-Stra{\ss}e 22,   D18059 Rostock, Germany\\
  ronny.tredup@uni-rostock.de
 }

 \runninghead{R. Devillers and R. Tredup}{Some Basic Techniques Allowing Petri Net Synthesis: Complexity...}

\maketitle

\begin{abstract}
In Petri net synthesis we ask whether a given transition system $A$ can be implemented by a Petri net $N$.
Depending on the level of accuracy, there are three ways how $N$ can implement $A$:
an \emph{embedding}, the least accurate implementation, preserves only the diversity of states of $A$;
a \emph{language simulation} already preserves exactly the language of $A$;
a \emph{realization}, the most accurate implementation, realizes the behavior of $A$ exactly.
However, whatever the sought implementation, a corresponding net does not always exist.
In this case, it was suggested to modify the input behavior -- of course as little as possible.
Since transition systems consist of states, events and edges, these components appear as a natural choice for modifications.
In this paper we show that the task of converting an unimplementable transition system into an implementable one by removing
as few states or events or edges as possible is NP-complete --regardless of what type of implementation we are aiming for;
we also show that the corresponding parameterized problems are $W[2]$-hard, where the number of removed components is considered as the parameter;
finally, we show there is no $c$-approximation algorithm (with a polynomial running time) for neither of these problems, for every constant $c\geq 1$.
\end{abstract}

\section{Introduction}\label{sec:introduction}%

Petri nets are a widely accepted framework for modeling and validating concurrent and distributed systems.
In general, there are two ways to deal with the behavior of Petri nets:
\emph{Analysis} starts from a given Petri net and investigates if its behavior satisfies some properties such as, for example, \emph{liveness},  \emph{reachability of deadlocks}, \emph{reachability of stable markings} or the \emph{fireability of a transition}~\cite{mcc:2017}.
In \emph{synthesis}, we deal with the opposite direction: starting from a regular behavior, given as a \emph{labeled transition system} (\LTS, for short), we try to find a Petri net that implements this behavior.

Synthesis of Petri nets has practical applications in numerous areas such as, for example, data and process mining~\cite{daglib/0027363,sac/PedroC16}, digital hardware design~\cite{644602,10.5555/2587933} and discovering of concurrency and distributability~\cite{fac/BadouelCD02,ershov/BestD11}.
On the other hand, Petri net synthesis has also been the subject of theoretical studies that, for example, aim at characterizing the complexity of synthesis~\cite{tapsoft/BadouelBD95} or look for structural properties
that classify an \LTS\  as implementable by subclasses of Petri nets such as, for example, \emph{marked graphs}, and thus allow improved synthesis procedures~\cite{acta/BestD15}.
It also led to various synthesis tools~\cite{BestS15,apt13,synet,CarmonaCK09,petrify}.

{\LTS}s have \emph{states}, \emph{events} and labeled \emph{edges}, i.e., ``source-event-target" triplets: the occurrence of the event at the source triggers a change of state to the target.
They have an \emph{initial state}, from which, triggered by an event-sequence, any other state is reachable.
Petri nets have \emph{places} containing \emph{tokens}, and an overall token distribution is considered as a \emph{marking}
 (i.e., a global state) of the net;
nets have \emph{transitions}, connected with places, which possibly can \emph{fire}:
the token distribution of their connected places may allow the firing.
A firing of a transition (locally) changes the token distribution (of its connected places) and thus (globally) the marking of the net.
They have an \emph{initial marking}, from which, triggered by the firing of a sequence of transitions, any other reachable marking is obtained.
The global behavior of a net is captured by its \emph{reachability graph}, which is a transition system, where reachable markings become states, transitions become events and edges correspond to ``marking-transition-marking"-triplets.
A Petri net $N$ implements an \LTS\  $A$, if the events of $A$ and the transitions of $N$ coincide and, moreover, if (the states of) $A$ and (the states of) the reachability graph of $N$ can be related by a mapping, which satisfies certain requirements.

According to the properties of the mappings, implementations with various degrees of ``accuracy" are possible:
Such a mapping is first required to be a \emph{simulation}, which means that every allowed sequence of events (starting at the initial state) can be simulated by a (fireable) sequence of transitions (starting at the initial marking).
However, finding a net that allows a simulation is not a challenge:
If $N$ is the net without places that has a transition $e$ for every event $e$ of $A$, then $N$ simulates $A$, since it can simply fire every sequence of events of $A$.
Moreover, this net can obviously simulate \emph{every} \LTS\  that has the same events as $A$.
From this point of view, $N$ simulates $A$ with the greatest inaccuracy and every information about the (forbidden) original behavior is lost.
On the other end of the spectrum, $N$ simulates $A$ most accurately if the simulation is an \emph{isomorphism}: then $N$ is an (exact) \emph{realization} of the behavior defined by $A$.
Unfortunately, not every \LTS\  can be realized by a Petri net.
However, this is actually not always necessary, depending on the application.
Therefore, \emph{embedding} and \emph{language simulation} have been discussed as other possible implementations in the literature, which --in a certain sense-- are less accurate, but still acceptable:
an embedding preserves at least the diversity of states, that is, the simulation map is injective;
a language simulation preserves exactly the allowed event sequences of $A$, that is, $A$ and $N$ are language-equivalent.
Unfortunately, although these \mbox{implementations} are less restrictive, they also do not ensure the existence of a corresponding net.
In order to achieve implementability, various techniques have been proposed in the literature that modify the components of the input behavior, i.e., its states, events and edges~\cite{txtcs/BadouelBD15,Verbeek:715007}.
One of the most discussed approaches among them is what is known as \emph{label-splitting}:
events are split into several (new) events and edges are relabeled so that edges that are initially labeled with the same event then are labeled with different events, which originate from the same event by splitting.
This method is relatively well-understood from the practical point of view --in the sense of available algorithms~\cite{topnoc/Carmona12,644602,topnoc/SchlachterW19}-- and from the theoretical point of view as well.
In particular, it was recently shown that achieving implementability by splitting as few events as possible is a problem hard to solve, namely NP-complete, regardless of the implementation kind~\cite{corr/abs-2002-04841,ictcs/Tredup20}.

In a natural way, the question arises whether there are other (small) suitable modifications to obtain implementable behaviors.
The answer is given by the nature of implementations itself:
States of the input are related with reachable markings of the implementing net, and if the behavior is not implementable, then for some of its states no reachable marking may exist.
Hence, the \emph{state removal} of such states (as well as the part of the {\LTS} which becomes unreachable) may lead to an implementable \LTS.
Occurrences of events at (source-) states of the input correspond to the firing of transitions in markings that are associated with the sources.
If the behavior is not implementable, then the firing of the transitions in the corresponding markings may not be possible.
Hence, removing such occurrences (i.e, the corresponding edges) may then yield an implementable behavior.
If the latter is an option, then there are two distinct ways to put the focus on the removal:
On the one hand, the modeler may allow the \emph{edge removal} of several edges that affect several events and, simultaneously, demand that some occurrences of every event remain --if this is possible.
On the other hand, the modeler may come to the conclusion that some events are less interesting than others and thus prefer the complete \emph{event removal} of the former (i.e., all of their occurrences) and the complete preservation of the latter.
Just like label-splitting, the removal of states, events and edges is a powerful transformation, since each of it is able to produce an implementable behavior:
For instance, when $A$ is degenerated to a single state, the resulting behavior is implementable.
Surely, however, this solution is not desirable.
Instead, we are interested in the corresponding optimization problems, that is, given an \LTS\  $A$, we are looking for a modification $A'$ of $A$, such that the number $\kappa$ of removed edges, events or states is as small as possible --depending on the technique applied.
If we turn the number $\kappa$ into a part of the input, then we obtain the corresponding decision version of the optimization problem.
Obviously, if we can solve the optimization problem, then we can solve the decision version as well (with only polynomial overhead).
Hence, the characterization of the computational complexity of the latter problem provides a lower bound of the complexity of the former.

In this paper, which is an extended version of~\cite{apn/Tredup21}, we completely characterize the computational complexity of \emph{state removal}, \emph{event removal} and \emph{edge removal} for all thinkable implementations they can aim for, i.e., embedding, language-simulation and realization.
In particular, we show that all of these decision problems are NP-complete and thus their optimization variant is also hard to solve.

Although the problem of finding an optimal edge, event, or state distance is difficult to solve, we still want to find such implementable modifications of \LTS.
There are at least two common ways to deal with such difficult problems:
First, one can investigate whether the decision variants can be solved exactly by an algorithm in which the exponential explosion of its running time depends functionally on only a (reasonable) parameter of the input, while the input length contributes only by a polynomial factor, i.e., one can ask whether the problem is \emph{fixed-parameter-tractable} (fpt for short).
In this paper, we show that, unfortunately, for all problems and all implementations this question is to be answered negatively
if we consider the number $\kappa$ of removed components as the parameter.
In particular, we show that all our problems parameterized by this natural parameter are at least $W[2]$-hard.

On the other hand, instead of an exact solution, one could also look for an approximate solution.
In this case we wish that such an approximate solution is provably of at least some constant quality.
That is, we look for a so-called c-approximation algorithm, where $c\geq 1$ is a constant, which for any given {\LTS} $A$ outputs in polynomial time an implementable \LTS\ $B$ which emerges from A by deleting at most $c$ times as many components as it would be optimally necessary.
In this paper, for all such problems, and all implementations, we show that such an algorithm does not exist, unless P=NP.

The paper is structured as follows:
Section~\ref{sec:prelis} provides the basic notions and supports them with some examples.
After that, Section~\ref{sec:edge_removal}, Section~\ref{sec:event_removal}, and Section~\ref{sec:state_removal} provide the NP-completeness of \emph{edge removal}, \emph{event removal}, and \emph{state removal}, respectively.
After that, in Section~\ref{sec:para_complex}, we argue that these problems are $W[2]$-hard when parameterized by the number of removed components.
In Section~\ref{sec:inapproximability}, we argue that none of the optimization versions of these problems allows for a c-approximation algorithm.
Finally, Section~\ref{sec:conclusion} briefly concludes and suggests subsequent possible developments.

\section{Preliminaries}\label{sec:prelis}%

This section provides the basic notions that we use throughout the paper and supports them with examples.

\subsection{Transition systems, and their modifications}\label{sec:ts}%

The overall starting point for the synthesis of Petri nets is a behavior that is given by a transition system:
\begin{definition}[Transition System]\label{def:transition_system}
A (deterministic, labeled) \emph{transition system} (LTS, for short) $A=(S,E, \delta,\iota)$ consists of two disjoint sets of \emph{states} $S$ and \emph{events}  $E$, a partial \emph{transition function} $\delta: S\times E \longrightarrow S$ and an \emph{initial state} $\iota\in S$.
An event $e$ \emph{occurs} at state $s$, denoted by $s\edge{e}$, if $\delta(s,e)$ is defined.
By $s\nedge{e}$ we denote that $\delta(s,e)$ is not defined.
We abridge $\delta(s,e)=s'$ by $s\edge{e}s'$ and call the latter an \emph{edge} with \emph{source} $s$ and \emph{target} $s'$.
By $s\edge{e}s'\in A$, we denote that the edge $s\edge{e}s'$ is present in $A$.
A sequence $s_0\edge{e_1}s_1, s_1\edge{e_2}s_2,\dots, s_{n-1}\edge{e_n}s_n$ of edges is called a (directed labeled) \emph{path} (from $s_0$ to $s_n$ in $A$).
$A$ is called \emph{reachable}, if there is a path from $\iota$ to $s$ for every state $s\in S$.
The \emph{language} of $A$ is the set of words $L(A)=\{e_1\dots e_n\in E^*\mid \exists s\in S: \iota\edge{e_1}\dots\edge{e_n}s\}\cup\{\varepsilon\}$, where $\varepsilon$ denotes the empty word.
\end{definition}

In the remainder of this paper, we always assume that LTSs are \emph{reachable}.
In this paper, we relate LTSs with the same set of events by so-called \emph{simulations}:

\begin{definition}
A \emph{simulation} between an \LTS\  $A=(S,E,\delta,\iota)$ and an \LTS\  $B=(S',E,\delta',\iota')$ is a mapping $\varphi: S\rightarrow S'$ such that $\varphi(\iota)=\iota'$ and if $s\edge{e}s'\in A$, then $\varphi(s)\edge{e}\varphi(s')\in B$;
$\varphi$ is called an \emph{embedding}, denoted by $A\hookrightarrow B$, if it is injective, that is, if $s\not=s'$, then $\varphi(s)\not=\varphi(s')$;
$\varphi$ is a \emph{language-simulation}, denoted by $A\triangleright B$, if $s\nedge{e}$ implies $\varphi(s)\nedge{e}$;
$\varphi$ is an \emph{isomorphism}, denoted by $A\cong B$, if it is bijective and $s\edge{e}s'\in A$ if and only if $\varphi(s)\edge{e}\varphi(s')\in B$.
\end{definition}

It is known from the literature that if $A\triangleright B$, then $L(A)=L(B)$~\cite{txtcs/BadouelBD15};
if $A\cong B$, then $A$ and $B$ are basically the same --but possibly for the names of their states.
An \LTS\  describes a behavior that is implementable or not.
In the latter case, we may apply the following \emph{modifications} in order to obtain an implementable LTS:

\begin{definition}[Edge Removal]\label{def:edge_removal}
Let $A=(S,E,\delta,\iota)$ be an \LTS.
An \LTS\  $B=(S', E',\delta',\iota)$ with state set $S'\subseteq S$ and event set $E'\subseteq E$ is an \emph{edge removal} of $A$ if, for all $e\in E'$ and all $s,s'\in S'$, holds: if $s\edge{e}s'\in B$, then $s\edge{e}s'\in A$.
By $\K=\{s\edge{e}s'\in A\mid s\edge{e}s'\not\in B\}$ we refer to the (set of) removed edges.
\end{definition}

\begin{definition}[Event Removal]\label{def:event_removal}
Let $A=(S,E,\delta,\iota)$ be an \LTS.
An \LTS\  $B=(S',E',\delta',\iota)$ with state set $S'\subseteq S$ and event set $E'\subseteq E$ is an \emph{event removal} of $A$ if for all $e\in E'$ the following is true:
$s\edge{e}s'\in B$ if and only if $s\edge{e}s'\in A$ for all $s,s'\in S$ .
By ${\E}=E\setminus E'$ we refer to the (set of) removed events.
\end{definition}

\begin{definition}[State Removal]
Let $A=(S,E,\delta,\iota)$ be an \LTS.
An \LTS\  $B=(S',E',\delta',\iota)$ with states $S'\subseteq S$ and events $E'\subseteq E$ is a \emph{state removal} of $A$ if the following two conditions are satisfied:
(1) $s\edge{e}s'\in B$ if and only if $s\edge{e}s'\in A$ for all $e\in E'$ and all $s,s'\in S'$;
(2) if $s\edge{e}s'\in A$ and $s\edge{e}s'\not\in B$, then $s\not\in S'$ or $s'\not\in S'$.
By ${\mS}=S\setminus S'$ we refer to the (set of) removed states.
\end{definition}

Notice that neither of these modifications is ``functional'', since, generally, there are several \LTS\  that can be considered as a suitable modification of $A$.
Moreover, edge removal is the most general modification introduced, since every event- or state removal is also an edge removal.
However, not every edge removal is an event removal or a state removal, not every event removal is a state removal, and not every state removal is an event removal.
In particular, there are substantial differences between these modifications that focus on different aspects of the LTS:
If $B$ is an edge removal, then there could possibly be an event $e\in E'$ for which there is an edge $s\edge{e}s'$ in $A$ that is not in $B$.
In contrast, if $B$ is an event removal and $e\in E'$, then every $e$-labeled edge of $A$ has to be present in $B$.
Furthermore, if $B$ is a state removal, then an edge $s\edge{e}s'$ of $A$ can only be missing in $B$ if its source $s$ or its target $s'$ is removed.
By contrast, the latter is not necessarily the case if $B$ is an event- or an edge removal.

\begin{example}\label{ex:modifications}
Consider the \LTS\  $A$ of Figure~\ref{fig:A}.
The \LTS\  $B$ of Figure~\ref{fig:state_removal} is a state removal of $A$ resulting by removing the state $s_3$, i.e., ${\mS}=\{s_3\}$.
$B$ is also an edge removal, where $\K=\{s_2\edge{x}s_3\}$.
However, this \LTS\  is not an event removal, since $x$ belongs to $B$, but not all $x$-labeled edges of $A$ are present.
The \LTS\  $C$ of Figure~\ref{fig:event_removal} is an event removal of $A$ such that ${\E}=\{a\}$.
$C$ is also an edge removal and $\K=\{t_0\edge{a}t_1, q_0\edge{a}q_1\}$, but is is not a state removal.
\end{example}

Note that, if $A$ is reachable, this is not always the case for $B$.
For instance, in Figure~\ref{fig:A}, if we remove event $x$, or state $s_1$ or edge $s_0\edge{x}s_1$, the result is not reachable.
In the following, we shall only consider reachable removals, however.
On the contrary, if $A$ is finite, so is $B$; in the following we shall only consider finite LTSs.

\begin{figure}[t!]
\begin{center}
\begin{minipage}{\textwidth}
\begin{center}
\begin{tikzpicture}[new set = import nodes]
\begin{scope}[nodes={set=import nodes}]%
		\node (bot) at (4.5,1) {\nscale{$\bot$}};
		\foreach \i in {0,1,2,3} { \coordinate (s\i) at (\i*1.5cm, 0) ;}
		\foreach \i in {0,1,2,3} { \node (s\i) at (s\i) {\nscale{$s_\i$}};}
		\foreach \i in {0,1} { \coordinate (t\i) at (\i*1.5cm+5.5cm, 0) ;}
		\foreach \i in {0,1} { \node (t\i) at (t\i) {\nscale{$t_\i$}};}
		\foreach \i in {0,1} { \coordinate (q\i) at (\i*1.5cm+8cm, 0) ;}
		\foreach \i in {0,1} { \node (q\i) at (q\i) {\nscale{$q_\i$}};}
\graph {(import nodes);
			s0 ->[pos=0.7,"\escale{$x$}"]s1->["\escale{$y$}"]s2->["\escale{$x$}"]s3;
			t0 ->["\escale{$x$}"]t1;
			t0 ->[ swap, bend right =40, "\escale{$a$}"]t1;
			q0 ->["\escale{$y$}"]q1;
			q0 ->[ swap, bend right =40, "\escale{$a$}"]q1;
			bot ->[swap, bend right =15, "\escale{$u$}"]s0;
			bot ->[ "\escale{$v$}"]t0;
			bot ->[bend left=15, "\escale{$w$}"]q0;
		};
\end{scope}
\end{tikzpicture}
\end{center}
\vspace{-2em}\caption{The \LTS\  $A$.
}\label{fig:A}
\end{minipage}

\vspace{1em}

\begin{minipage}{\textwidth}
\begin{center}
\begin{tikzpicture}[new set = import nodes]
\begin{scope}[nodes={set=import nodes}]%
		\node (bot) at (4.5,1) {\nscale{$\bot$}};
		\foreach \i in {0,1,2} { \coordinate (s\i) at (\i*1.5cm, 0) ;}
		\foreach \i in {0,1,2} { \node (s\i) at (s\i) {\nscale{$s_\i$}};}
		\foreach \i in {0,1} { \coordinate (t\i) at (\i*1.5cm+5.5cm, 0) ;}
		\foreach \i in {0,1} { \node (t\i) at (t\i) {\nscale{$t_\i$}};}
		\foreach \i in {0,1} { \coordinate (q\i) at (\i*1.5cm+8cm, 0) ;}
		\foreach \i in {0,1} { \node (q\i) at (q\i) {\nscale{$q_\i$}};}
\graph {(import nodes);
			s0 ->[pos=0.7, "\escale{$x$}"]s1->["\escale{$y$}"]s2;
			t0 ->["\escale{$x$}"]t1;
			t0 ->[ swap, bend right=40, "\escale{$a$}"]t1;
			q0 ->["\escale{$y$}"]q1;
			q0 ->[swap,  bend right =40, "\escale{$a$}"]q1;
			bot ->[swap, bend right =15, "\escale{$u$}"]s0;
			bot ->[ "\escale{$v$}"]t0;
			bot ->[bend left=15, "\escale{$w$}"]q0;
		};
\end{scope}
\end{tikzpicture}
\end{center}
\vspace{-2em}\caption{The state removal $B$ of $A$ that results by removing the state $s_3$.
}\label{fig:state_removal}
\end{minipage}

\vspace{1em}

\begin{minipage}{\textwidth}
\begin{center}
\begin{tikzpicture}[new set = import nodes]
\begin{scope}[nodes={set=import nodes}]%
		\node (bot) at (4.5,1) {\nscale{$\bot$}};
		\foreach \i in {0,1,2,3} { \coordinate (s\i) at (\i*1.5cm, 0) ;}
		\foreach \i in {0,1,2,3} { \node (s\i) at (s\i) {\nscale{$s_\i$}};}
		\foreach \i in {0,1} { \coordinate (t\i) at (\i*1.5cm+6cm, 0) ;}
		\foreach \i in {0,1} { \node (t\i) at (t\i) {\nscale{$t_\i$}};}
		\foreach \i in {0,1} { \coordinate (q\i) at (\i*1.5cm+9cm, 0) ;}
		\foreach \i in {0,1} { \node (q\i) at (q\i) {\nscale{$q_\i$}};}
\graph {(import nodes);
			s0 ->[pos=0.7,"\escale{$x$}"]s1->["\escale{$y$}"]s2->["\escale{$x$}"]s3;
			t0 ->["\escale{$x$}"]t1;
			q0 ->["\escale{$y$}"]q1;
			bot ->[swap, bend right =15, "\escale{$u$}"]s0;
			bot ->[ "\escale{$v$}"]t0;
			bot ->[bend left=15, "\escale{$w$}"]q0;
		};
\end{scope}
\end{tikzpicture}
\end{center}
\vspace{-1em}\caption{The event removal $C$ of $A$ that results by removing the event $a$.
}\label{fig:event_removal}
\end{minipage}
\end{center}\vspace*{-4mm}
\end{figure}

\subsection{Petri nets, and implementations}\label{sec:pn}%

Petri nets are the target model with which we want to implement \LTS:
\begin{definition}[Petri Nets]\label{def:petri_nets}
A \emph{(weighted) Petri net} $N=(P,T,f,M_0)$ consists of finite and disjoint sets of \emph{places} $P$ and \emph{transitions} $T$, a (total) \emph{flow} $f: ((P \times T) \cup (T \times P)) \rightarrow \mathbb{N}$ and an \emph{initial marking} $M_0: P \rightarrow \mathbb{N}$.
A transition $t\in T$ can \emph{fire} or \emph{occur} in a marking $M:P\rightarrow \mathbb{N}$, denoted by $M\edge{t}$, if $M(p)\geq f(p,t) $ for all places $p\in P$.
The firing of $t$ in marking $M$ leads to the marking $M'(p)=M(p)-f(p,t)+f(t,p)$ for all $p\in P$, denoted by $M\edge{t}M'$.
This notation extends to sequences $w \in T^*$ and the \emph{reachability set} $RS(N)=\{M \mid \exists w\in T^*: M_0\edge{w}M \}$ contains all the reachable markings of $N$.
The \emph{reachability graph} of $N$ is the \LTS\  $A_N=(RS(N), T,\delta, M_0)$, where, for every reachable marking $M$ of $N$ and transition $t \in T$, the transition function $\delta$ of $A_N$ is defined by $\delta(M,t) = M'$ if and only if $M \edge{t} M'$.
\end{definition}

Simulations between $A$ and $A_N$ define how a net $N$ implements an \LTS\  $A$:
\begin{definition}[Implementations]\label{def:implementation}
If $A$ is an \LTS\  and $N$ is a Petri net, then $N$ is an \emph{embedding} of $A$ if $A\hookrightarrow A_N$;
$N$ is a \emph{language-simulation} of $A$, if $A\triangleright A_N$, and
$N$ is a \emph{realization} of $A$, if $A\cong A_N$.
We say $N$ \emph{implements} $A$, if it is an embedding or a language-simulation or a realization of $A$.
\end{definition}

\subsection{Regions, separation properties, and synthesized nets}\label{sec:regions}%

If a Petri net $N$ implements an \LTS\  $A$, then the events of $A$ are the transitions of $N$.
We obtain the remaining components of $N$, that is, places, flow and initial marking, by regions of $A$:
\begin{definition}[Region]\label{def:region}
A \emph{region} $R=(sup, con, pro)$ of an \LTS\  $A=(S, E, \delta, \iota)$ consists of the mappings \emph{\underline{sup}port} $sup:S \rightarrow \mathbb{N}$ and \emph{\underline{con}sume} and \emph{\underline{pro}duce} $con, pro:E \rightarrow \mathbb{N}$ such that if $s \edge{e} s'$ is an edge of $A$, then $con(e)\leq sup(s)$ and $sup(s')=sup(s)-con(e)+pro(e)$.
\end{definition}

\begin{remark}\label{rem:implicit}
A region $R=(sup, con, pro)$ is \emph{implicitly} completely defined by $sup(\iota)$, $con$ and $pro$:
Since $A$ is reachable, there is a path $\iota\edge{e_1}\dots \edge{e_n}s_n$ such that $s=s_n$  for every state $s\in S$.
Consequently, we inductively obtain $sup(s_{i+1})$ by $sup(s_{i+1})=sup(s_i) -con(e_{i+1}) +pro(e_{i+1})$ for all $i\in \{0,\dots, n-1\}$ and $s_0 = \iota$.
Hence, for the sake of simplicity, we often present regions only implicitly, since $sup$ and thus $R$ can be obtained from $sup(\iota)$, $con$ and $pro$ (one has to check however that, if two edges lead to the same state, the yielded support is the same).
For an even more compact presentation, for $c,p\in\mathbb{N}$, we group events with the same ``behavior'' together by $\T_{c,p}^R=\{e\in E\mid con(e)=c\text{ and } pro(e)=p\}$.
\end{remark}

\begin{definition}[Effect]\label{def:effect}
Let $A=(S,E,\delta, \iota)$ be a LTS, and $(sup, con, pro)$ be a region of $A$.
If $e\in E$, then we say that $\eff(e)= - con(e) + pro(e)$ (in $\mathbb{Z}$) is the \emph{effect} of $e$.
\end{definition}

\begin{lemma}\label{lem:effect_on_a_path}
Let $A=(S,E,\delta, \iota)$ be a LTS, and $(sup, con, pro)$ be a region of $A$.
If $s_0\edge{e_1}\dots\edge{e_n}s_n$ is a path in $A$, then it holds $sup(s_n)=sup(s_0)+\sum_{i=0}^n \eff(e_i)$.
\end{lemma}
\begin{proof}
Since $R$ is a region, we have $sup(s_{i+1})=sup(s_i)-con(e_{i+1})+pro(e_{i+1})$, and thus $sup(s_{i+1})=sup(s_i)+\eff(e_{i+1})$ for all $i\in \{0,\dots, n-1\}$.
Hence, the claim follows easily by induction.
\end{proof}

If there is an implementing net $N$ for $A$, then each place correspond to a region $R=(sup, con, pro)$ of $A$:
$con(e)$ and $pro(e)$ model $f(R,e)$ and $f(e,R) $ for all transitions $e$, respectively, and $sup(\iota)$ models the initial marking $M_0(R)$.
In particular, every set of regions defines a \emph{synthesized net}:
\begin{definition}[Synthesized Net]\label{def:synthesized_net}
A set $\mathcal{R}$ of regions of \LTS\  $A=(S,E,\delta, \iota)$ defines the \emph{synthesized net} $N^{\mathcal{R}}_A=(\mathcal{R}, E, f, M_0)$, where  $f(R,e)=con(e)$ and $f(e,R)=pro(e)$ and $M_0(R)=sup(\iota)$ for all $R=(sup, con,pro)\in \mathcal{R}$ and $e\in E$.
\end{definition}

If the synthesized net is an embedding or a realization of $A$, then distinct states of $A$ correspond to distinct markings of the net.
The net $N_A^{\mathcal{R}}$ satisfies this requirement if the set $\mathcal{R}$ of regions present the \emph{state separation property}:
\begin{definition}[State Separation Property]\label{def:ssp}
A pair $(s, s')$ of distinct states of \LTS\  $A=(S,E,\delta, \iota)$ defines a \emph{states separation atom} (SSA).
A region $R=(sup, con, pro)$ \emph{solves} $(s,s')$ if $sup(s)\not=sup(s')$.
We say a state $s$ is \emph{solvable} if, for every $s'\in S\setminus \{s\}$, there is a region that solves the SSA $(s,s')$.
If every SSA or, equivalently, every state of $A$ is solvable, then $A$ has the \emph{state separation property} (SSP).
\end{definition}

If the net is a language-simulation or a realization, then the firing of a transition $e$ must be inhibited at a marking $M$ whenever the event $e$ does not occur at the state $s$ that correspond to $M$ via $\varphi$.
This is ensured if $\mathcal{R}$ witnesses the \emph{event/state separation property}:
\begin{definition}[Event/State Separation Property]\label{def:essp}
A pair $(e,s)$ of event $e$ and state $s$ of \LTS\  $A=(S,E,\delta, \iota)$ such that $s\nedge{e}$ defines an \emph{event/state separation atom} (ESSA).
A region $R=(sup, con, pro)$ \emph{solves} $(e,s)$ if $sup(s) < con(e) $.
We say an event $e$ is \emph{solvable} if, for every $s\in S$ with $s\nedge{e}$, there is a region that solves the ESSA $(e,s)$.
If every ESSA or, equivalently, every event of $A$ is solvable, then $A$ has the \emph{event/state separation property} (ESSP).
\end{definition}

A set $\mathcal{R}$ of regions of $A$ is called a \emph{witness} for the SSP or the ESSP (of $A$) if, for every SSA or ESSA, there is a region in $\mathcal{R}$ that solves it.
The next lemma is based on~\cite[p.~162]{txtcs/BadouelBD15} and~\cite[p.~214 ff.]{txtcs/BadouelBD15} and states in which case the existence of a witness and the existence of an implementation are equivalent; this will allow us to formulate our decision problems rather on the notion of witnesses than on the notion of implementations (notice that Petri nets correspond to the type of nets $\tau_{PT}$ in \cite[p.~130]{txtcs/BadouelBD15}):
\begin{lemma}[\cite{txtcs/BadouelBD15}]\label{lem:badouel}
Let $A$ be an \LTS\  and $N$ a Petri net.
\begin{enumerate}
\item
$A\hookrightarrow A_N$ if and only if there is a witness $\mathcal{R}$ for the SSP of $A$ and $N=N^{\mathcal{R}}_A$;
\item
$A\triangleright A_N$ if and only if there is a witness $\mathcal{R}$ for the ESSP of $A$ and $N=N^{\mathcal{R}}_A$;
\item
$A\cong A_n$  if and only if there is a witness $\mathcal{R}$ for both the SSP and the ESSP of $A$ and $N=N^{\mathcal{R}}_A$.
\item
Whether $A$ has the SSP or the ESSP can be decided and, in case of a positive decision, a witness can be computed in polynomial time.
\end{enumerate}
\end{lemma}

\begin{example}\label{ex:embedding}
Let $A=(Z,E,\delta,\bot)$ be the \LTS\  of Figure~\ref{fig:A}.
The following implicitly defined region $R=(sup, con, pro)$ solves all SSA of $A$:
$sup(\bot)=8$ and $\T_{5,0}^{R}=\{v\}$, $\T_{7,0}^{R}=\{w\}$ and $\T_{1,0}^{R}=E\setminus \{v,w\}$.
According to Remark~\ref{rem:implicit}, one obtains $R$ explicitly: $sup(s_i)=7-i$ for all $i\in \{0,\dots,3\}$ and $sup(t_0)=3$, $sup(t_1)=2$, $sup(q_0)=1$ and $sup(q_1)=0$.
The set $\mathcal{R}=\{R\}$ witnesses the SSP of $A$ and the net $N=N_A^{\mathcal{R}}$ is an embedding of $A$.
Figure~\ref{fig:net} shows $N$ (top) and its reachability graph $A_N$ (bottom).
The injective simulation map $\varphi$ is defined by $\varphi(\bot)=8$, $\varphi(s_i)=7-i$ for all $i\in \{0,\dots,3\}$ and $\varphi(t_0)=3$, $\varphi(t_1)=2$, $\varphi(q_0)=1$ and $\varphi(q_1)=0$.
\end{example}

\begin{example}\label{ex:not_essp}
The \LTS\  $A=(Z,E,\delta,\bot)$ of Figure~\ref{fig:A} does not have the ESSP, since the ESSA $\alpha=(x,s_1)$ is not solvable.
This can be seen as follows:
Assume $R=(sup, con,pro)$ is a region that solves $\alpha$, that is, $con(x)> sup(s_1)$.
Since $x$ occurs at $s_0$, we have $con(x)\leq sup(s_0)$.
By $sup(s_1)=sup(s_0)-con(x)+pro(x)$ and $con(x)> sup(s_1)$, this implies $con(x)> pro(x)$.
By $t_0\edge{x}t_1$, this also implies $sup(t_0) > sup(t_1)$ and thus $con(a)> pro(a)$ by $t_0\edge{a}t_1$.
On the other hand, $x$ occurs at $s_2$, which implies $con(x)\leq sup(s_2)$.
By $con(x)>sup(s_1)$ and $s_1\edge{y}s_2$, this is only possible if $con(y) < pro(y)$, which implies $sup(q_0) < sup(q_1)$ and thus $con(a) < pro(a)$.
This is a contradiction.
Hence, $\alpha$ is not solvable.
\end{example}

\begin{example}\label{ex:removals}
The \LTS\  of Figure~\ref{fig:A} does not have the ESSP, since the ESSA $\alpha=(x,s_1)$ is not solvable, as we just saw.
However, for the \LTS\  $B$ of Figure~\ref{fig:state_removal}, which is a state removal for $A$, there is a region $R=(sup, con, pro)$ that solves the ESSA $(x,s_1)$, which is implicitly defined as follows:
$sup(\bot)=2$ and $\T_{2,1}^{R}=\{x\}$ and $\T_{1,0}^{R}=\{a,y\}$ and $\T_{0,0}^{R}=E\setminus \{a,y,x\}$.
One finds out that the remaining ESSA of $B$ are also solvable.
Hence, $B$ has the ESSP and the SSP.
If the modeler comes to the conclusion that the events $x$ and $y$, their corresponding edges and all of their sources and targets are essential for the modeled behavior, then a realizable behavior can also be obtained as the event removal $C$ of $A$ as defined in Figure~\ref{fig:event_removal}.
A Region $R'=(sup', con', pro')$ solving $(x,s_1)$ in $C$, is then implicitly defined as follows:
$sup(\bot)=1$ and $\T_{1,0}^{R'}=\{x\}$ and $\T_{0,1}^{R'}=\{y\}$ and $\T_{0,0}^{R'}=E\setminus \{x,y\}$.
\end{example}

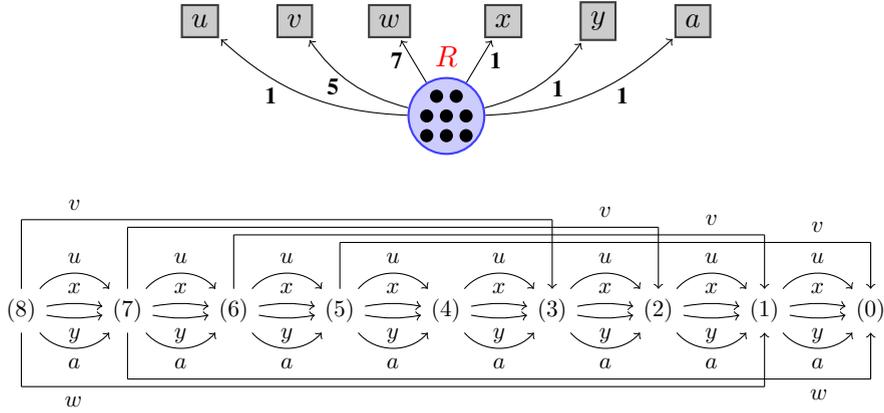
\begin{figure}[t!]
\begin{center}
\begin{minipage}{\textwidth} 
\begin{center}
\tikzstyle{place}=[circle,thick,draw=blue!75,fill=blue!20,minimum size=10mm]
\tikzstyle{red place}=[place,draw=red!75,fill=red!20]
\tikzstyle{transition}=[rectangle,thick,draw=black!75,
  			  fill=black!20,minimum size=4mm]
			   \tikzstyle{every label}=[red]
\begin{tikzpicture}[new set = import nodes]
\begin{scope}[node distance=1.25cm,bend angle=45,auto]
\node (m) at (0,0) {};
\node [transition, node distance=0.75cm] (w)[left of=m]{$w$};
\node [transition] (v)[left of=w]{$v$};
\node [transition] (u)[left of=v]{$u$};
\node [transition, node distance=0.75cm] (x)[right of=m]{$x$};
\node [transition] (y)[right of=x]{$y$};
\node [transition] (a)[right of=y]{$a$};
\node [place] (R)[tokens=8, below of=m, label=above: $R$]{}
      edge [post,left] node {$\nscale{7}$}  (w)
      edge [post,below, pos=0.7, bend left= 20] node {$\nscale{5}$} (v)
      edge [post,below, pos=0.7, bend left= 20] node {$\nscale{1}$}  (u)
      edge [post,right] node {$\nscale{1}$}  (x)
      edge [post,below, pos=0.7, bend right= 20] node {$\nscale{1}$}  (y)
      edge [post,below, pos=0.7, bend right= 20] node {$\nscale{1}$}  (a);
\end{scope}
\end{tikzpicture}
\end{center}

\begin{center}
\begin{tikzpicture}[new set = import nodes]
\begin{scope}[yshift=-3cm, nodes={set=import nodes}]
		\foreach \i in {8,...,0} { \coordinate (\i) at (11.2cm-\i*1.4cm, 0) ;}
		\foreach \i in {8,...,0} { \node (\i) at (\i) {\nscale{$(\i)$}};}
		\coordinate (h00) at (0,1.2);
		\coordinate (h01) at (7,1.2);
		\coordinate (h10) at (1.4,1.1);
		\coordinate (h11) at (8.4,1.1);
		\coordinate (h20) at (2.8,1);
		\coordinate (h21) at (9.8,1);
		\coordinate (h30) at (4.2,0.9);
		\coordinate (h31) at (11.2,0.9);
		\coordinate (g00) at (0,-1);
		\coordinate (g01) at (9.8,-1);
		\coordinate (g10) at (1.4,-0.9);
		\coordinate (g11) at (11.2,-0.9);
\graph {(import nodes);
	
			8 ->[bend left=10, "\escale{$x$}"]7 ->[bend left=10, "\escale{$x$}"]6 ->[bend left=10, "\escale{$x$}"]5 ->[bend left=10, "\escale{$x$}"]4 ->[bend left=10, "\escale{$x$}"]3 ->[bend left=10, "\escale{$x$}"]2 ->[bend left=10, "\escale{$x$}"]1 ->[bend left=10, "\escale{$x$}"]0;
			8 ->[bend left=50, "\escale{$u$}"]7->[bend left=50, "\escale{$u$}"]6->[bend left=50, "\escale{$u$}"]5->[bend left=50, "\escale{$u$}"]4->[bend left=50, "\escale{$u$}"]3->[bend left=50, "\escale{$u$}"]2->[bend left=50, "\escale{$u$}"]1->[bend left=50, "\escale{$u$}"]0;
			8 ->[swap, bend right=10, "\escale{$y$}"]7 ->[swap, bend right=10, "\escale{$y$}"]6 ->[swap, bend right=10, "\escale{$y$}"]5 ->[swap, bend right=10, "\escale{$y$}"]4 ->[swap, bend right=10, "\escale{$y$}"]3 ->[swap, bend right=10, "\escale{$y$}"]2 ->[swap, bend right=10, "\escale{$y$}"]1 ->[swap, bend right=10, "\escale{$y$}"]0;
			8 ->[swap, bend right=50, "\escale{$a$}"]7->[swap, bend right=50, "\escale{$a$}"]6->[swap, bend right=50, "\escale{$a$}"]5->[swap, bend right=50, "\escale{$a$}"]4->[swap, bend right=50, "\escale{$a$}"]3->[swap, bend right=50, "\escale{$a$}"]2->[swap, bend right=50, "\escale{$a$}"]1->[swap, bend right=50, "\escale{$a$}"]0;
			8--h00--[pos=0.1,"\escale{$v$}"]h01->3;
			7--h10--[pos=0.9,"\escale{$v$}"]h11->2;
			6--h20--[pos=0.9,"\escale{$v$}"]h21->1;
			5--h30--[pos=0.9,"\escale{$v$}"]h31->0;
			8--g00--[swap,pos=0.07, "\escale{$w$}"]g01->1;
			7--g10--[pos=0.93,swap, "\escale{$w$}"]g11->0;
			};
\end{scope}
\end{tikzpicture}
\end{center}\vspace*{-6mm}
\caption{The net $N=N_A^{\mathcal{R}}$ and its reachability graph $A_N$ according to Example~\ref{ex:embedding}.}\label{fig:net}
\end{minipage}
\end{center}\vspace*{-5mm}
\end{figure}

\section{The complexity of  Edge Removal}\label{sec:edge_removal}%

According to Lemma~\ref{lem:badouel}, the question whether a particular implementation for a given \LTS\  exists is equivalent to the question whether the \LTS\  has the separation properties that correspond to the implementation.
In this section, we are interested in modifying an \LTS\  into an implementable one by the removal of a bounded number of edges.
In particular, we are interested in the computational complexity of the following decision problems:

\noindent
\fbox{\begin{minipage}[t][1.8\height][c]{0.97\textwidth}
\begin{decisionproblem}
  \problemtitle{\textsc{Edge Removal for Embedding}}
  \probleminput{An \LTS\  $A=(S,E,\delta,\iota)$, a natural number $\kappa$.}
  \problemquestion{Does there exist an edge removal $B$ for $A$ by $\K$ that has the SSP and satisfies $\vert\K\vert\leq \kappa$?}
\end{decisionproblem}
\end{minipage}}
\smallskip

\noindent
\fbox{\begin{minipage}[t][1.8\height][c]{0.97\textwidth}
\begin{decisionproblem}
  \problemtitle{\textsc{Edge Removal for Language-Simulation}}
  \probleminput{An \LTS\  $A=(S,E,\delta,\iota)$, a natural number $\kappa$.}
  \problemquestion{Does there exist an edge removal $B$ for $A$ by $\K$ that has the ESSP and satisfies $\vert\K\vert\leq \kappa$?}
\end{decisionproblem}
\end{minipage}}
\smallskip

\noindent
\fbox{\begin{minipage}[t][1.8\height][c]{0.97\textwidth}
\begin{decisionproblem}
  \problemtitle{\textsc{Edge Removal for Realization}}
  \probleminput{An \LTS\  $A=(S,E,\delta,\iota)$, a natural number $\kappa$.}
  \problemquestion{Does there exist an edge removal $B$ for $A$ by $\K$ that has the ESSP and the SSP and satisfies $\vert\K\vert\leq \kappa$?}
\end{decisionproblem}
\end{minipage}}
\smallskip

\subsection{Edge Removal aiming at language-simulation or realization}\label{sec:edge_removal_realization}%

The following theorem characterizes the complexity of both {Edge Removal for Language-simulation}, and {Edge Removal for realization}:
\begin{theorem}\label{the:edge_removal_realization}
{Edge Removal for Language-simulation} as well as {Edge Removal for realization} are NP-complete.
\end{theorem}

It is easy to see that the addressed problems are in NP:
If there is an adequate edge removal for $A$, then a Turing machine can guess $\K$ by a non-deterministic computation in a time polynomial in the size of the input.
After that, the machine can deterministically (and polynomially)
compute $B$ and, since the size of $B$ is bounded by the size of $A$, it can compute a witness for the relevant property of $B$ in a time polynomial in the size of the input by Lemma~\ref{lem:badouel}.
Hence, in order to complete the NP-completeness part of Theorem~\ref{the:edge_removal_realization}, it remains to prove the NP-hardness of the problems.
This proof is based on a polynomial reduction of the problem \textsc{Hitting Set}:

\noindent
\fbox{\begin{minipage}[t][1.7\height][c]{0.97\textwidth}
\begin{decisionproblem}
  \problemtitle{\textsc{Hitting Set (HS)}}
  \probleminput{A triple $({\U},M,\lambda)$ that consist of a finite set ${\U}$, a set $M=\{M_0,\dots, M_{m-1}\}$ of subsets of ${\U}$ and a natural number $\lambda$.}
  \problemquestion{Does there exist a hitting set $Z$ for $({\U},M)$, that is, $Z \subseteq {\U}$ and $Z \cap M_i\not=\emptyset$ for all $i\in \{0,\dots, m-1\}$, that satisfies $\vert Z \vert \leq \lambda$?}
\end{decisionproblem}
\end{minipage}}

\begin{example}\label{ex:hs}
The instance $({\U}, M, 4)$ such that ${\U}=\{X_0,\dots,X_5\}$ and $M=\{M_0,\dots, M_8\}$, where
$M_0=\{X_0,X_1\}$,
$M_1=\{X_0,X_3\}$,
$M_2=\{X_0,X_5\}$,
$M_3=\{X_1,X_2\}$,
$M_4=\{X_1,X_5\}$,
$M_5=\{X_2,X_3\}$,
$M_6=\{X_2,X_4\}$,
$M_7=\{X_3,X_4\}$,
$M_8=\{X_4,X_5\}$, has, for example, the hitting set
$Z=\{X_0,X_2,X_3,X_5\}$ and thus allows a positive decision.
\end{example}

Without loss of generality, we may assume that $\lambda \leq \vert \U\vert$, and, for all $i\in \{0,\dots, m-1\}$, $|{\U}|>|M_i|>1$.
Indeed, if $M_i=\emptyset$ we know there is no solution,  if $M_i=\{x\}$ we know $x$ must be in $Z$ and the problem may be reduced to another, smaller, one.

 \begin{theorem}[\cite{coco/Karp72}]\label{the:hitting_set_np}
{Hitting Set} is NP-complete.
\end{theorem}

In the following, unless explicitly stated otherwise, let $({\U}, M, \lambda)$ be an arbitrary but fixed input of \textsc{HS} such that  ${\U}=\{X_0,\dots, X_{n-1}\}$ and $M=\{M_0,\dots,M_{m-1}\}$, where $M_i=\{X_{i_0},\dots, X_{i_{m_i-1}}\}$ (and thus $\vert M_i\vert =m_i$) for all $i\in \{0,\dots, m-1\}$.
For technical reasons, we assume without loss of generality that $i_0 <\dots < i_{m_i-1}$ for the elements  $X_{i_0},\dots, X_{i_{m_i-1}}$ of the set $M_i$ for all $i\in \{0,\dots, m-1\}$.

\paragraph{The General Idea of the Reduction.}
We reduce $(\U,M,\lambda)$ to a pair $(A,\kappa)$ of \LTS\  $A$ and a natural number $\kappa$ in such a way that the following implications are satisfied:
If there is an edge removal $B$ of $A$ that removes at most $\kappa$ edges and has the ESSP, then there is a hitting set with at most $\lambda$ elements for $(\U,M)$.
Conversely, if there is a hitting set with $\lambda$ elements for $(\U,M)$, then there is an edge removal $B$ of $A$ that removes at most $\kappa$ edges and has both the ESSP and the SSP.
Notice that such a reduction proves the hardness of both \textsc{Edge Removal for Language-Simulation} and \textsc{Edge Removal for Realization}.
The announced \LTS\ $A$ consists of several components.
Just as it is common in the world of reductions, we refer to these components as \emph{gadgets}.

\paragraph{The Reduction.}
For a start, we define $\kappa=\lambda$.
Moreover, for every $i\in \{0,\dots, m-1\}$, and for every $j\in \{0,\dots, \kappa\}$, the \LTS\  $A$ has the following gadget $T_{i,j}$, that uses the elements of $M_i=\{ X_{i_0},\dots, X_{m_i-1}\}$ as events:

\begin{center}
\begin{tikzpicture}[new set = import nodes]
\begin{scope}[nodes={set=import nodes}]
	
	\node (t) at (-0.9,0){$T_{i,j}=$};
	\foreach \i in {0,...,2} {\coordinate (\i) at (\i*2cm,0);}
	\foreach \i in {3} {\coordinate (\i) at (\i*2cm+6,0);}
	\foreach \i in {4} {\coordinate (\i) at (8.75,0);}
	\foreach \i in {0,1} {\node (\i) at (\i) {\nscale{$t_{i,j,\i}$}};}
	\node (2) at (2) {$\dots$};
	\node (3) at (3) {\nscale{$t_{i,j,m_i}$}};
	\node (4) at (4) {\nscale{$t_{i,j,m_i+1}$}};
	\graph {
	(import nodes);
			0->["\escale{$X_{i_0}$}"]1->["\escale{$X_{i_1}$}"]2->["\escale{$X_{i_{m_i-1}}$}"]3->["\escale{$X_{i_0}$}"]4;
			
			};
\end{scope}
\end{tikzpicture}
\end{center}

Notice that $X_{i_0}$ is both the first, and the last event along the path $T_{i,j}$, while all the other events of $M_i\setminus\{X_{i_0}\}$ occur exactly once.
Moreover, since there are essentially $\kappa+1$ copies of the same sequence of events $X_{i_0}\dots X_{m_i-1}X_{i_0}$, after the removal of at most $\kappa$ edges, there is a least one of these sequences left.\\

For every $i\in \{0,\dots, n-1\}$, the \LTS\  has the following gadget $F_i$ that has $X_i$ as event, as well as an $a_\ell$-labeled edge, for all $\ell\in \{0,\dots, \kappa\}$, with the same direction:

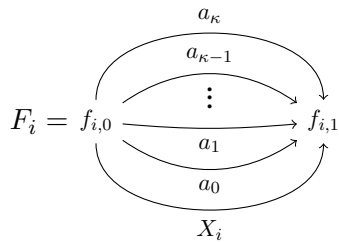
\begin{center}
\begin{tikzpicture}[new set = import nodes]
\begin{scope}[xshift=7.5cm, nodes={set=import nodes}]
	\node (F) at (-0.75, 0) {$F_i=$};
	\coordinate (0) at (0,0);
	\coordinate (1) at (3,0);
	\node (0) at (0) {\nscale{$f_{i,0}$}};
	\node (1) at (1) {\nscale{$f_{i,1}$}};
	\node (dots) at (1.5,0.4) {\nscale{$\vdots$}};
	\graph {
	(import nodes);
			
			0->[ swap, bend right=90,  "\escale{$X_i$}"]1;
			0->[swap, bend right =35, "\escale{$a_0$}"]1;
			0->[ bend right  =5, swap,"\escale{$a_1$}"]1;
			0->[swap, bend left =35,swap , "\escale{$a_{\kappa-1}$}"]1;
			0->[swap, bend left =90, swap, "\escale{$a_\kappa$}"]1;

			};
\end{scope}
\end{tikzpicture}
\end{center}

Finally, the \LTS\ $A$ has the initial state $\iota$, and uses the following edges to connect  $\iota$ to the just introduced gadgets:
\begin{itemize}
\itemsep=0.9pt
\item
For all $i\in \{0,\dots, m-1\}$, and all $j\in \{0,\dots, \kappa\}$, the \LTS\  $A$ has the edge $\iota\Edge{u_i^j}t_{i,j,0}$.
\item
For all $i\in \{0,\dots, n-1\}$, the \LTS\  $A$ has the edge $\iota\Edge{v_i}f_{i,0}$.
\end{itemize}

The resulting \LTS\  is $A=(S,E,\delta,\iota)$.
It is easy to check that $A$ is reachable.
Figure~\ref{fig:reduction_edge_removal} provides (essentially) the gadgets of the LTS $A$ that are based on the input
of Example~\ref{ex:hs}.
In what follows, we prove that $A$ satisfies the conditions stated above in the description of the general idea of reduction.
(Recall that $\K$ refers to the set of edges removed from $A$, where the edge removal is clear from the context).

\begin{figure}[htb]
\vspace*{-1mm}
\begin{center}
\begin{tikzpicture}[new set = import nodes]
\begin{scope}[nodes={set=import nodes}]

		\foreach \j / \i in {
		0/0,0/1,0/2,0/3,  1/0,1/1,1/2,1/3,  2/0,2/1,2/2,2/3,  3/0,3/1,3/2,3/3,  4/0,4/1,4/2,4/3} {\coordinate (t\j\i) at (\i*1.5cm,-\j*1cm);}
	\foreach \j / \i in {
	0/0,0/1,0/2,0/3,  1/0,1/1,1/2,1/3,  2/0,2/1,2/2,2/3,  3/0,3/1,3/2,3/3,  4/0,4/1,4/2,4/3}  {\node (t\j\i) at (t\j\i) {\nscale{$t_{\j,j,\i}$}};}
	\graph {
	(import nodes);
		t00->["\escale{$X_0$}"]t01->["\escale{$X_1$}"]t02->["\escale{$X_0$}"]t03;
		t10->["\escale{$X_0$}"]t11->["\escale{$X_3$}"]t12->["\escale{$X_0$}"]t13;
		t20->["\escale{$X_0$}"]t21->["\escale{$X_5$}"]t22->["\escale{$X_0$}"]t23;
		t30->["\escale{$X_1$}"]t31->["\escale{$X_2$}"]t32->["\escale{$X_1$}"]t33;
		t40->["\escale{$X_1$}"]t41->["\escale{$X_5$}"]t42->["\escale{$X_1$}"]t43;
			};
\end{scope}
\begin{scope}[xshift=6cm,nodes={set=import nodes}]

		\foreach \j / \i in {
		5/0,5/1,5/2,5/3,  6/0,6/1,6/2,6/3,  7/0,7/1,7/2,7/3,  8/0,8/1,8/2,8/3} {\coordinate (t\j\i) at (\i*1.5cm,-\j*1cm+5cm);}
	\foreach \j / \i in {
	5/0,5/1,5/2,5/3,  6/0,6/1,6/2,6/3,  7/0,7/1,7/2,7/3,  8/0,8/1,8/2,8/3}  {\node (t\j\i) at (t\j\i) {\nscale{$t_{\j,j,\i}$}};}
	\graph {
	(import nodes);
		t50->["\escale{$X_2$}"]t51->["\escale{$X_3$}"]t52->["\escale{$X_2$}"]t53;
		t60->["\escale{$X_2$}"]t61->["\escale{$X_4$}"]t62->["\escale{$X_2$}"]t63;
		t70->["\escale{$X_3$}"]t71->["\escale{$X_4$}"]t72->["\escale{$X_3$}"]t73;
		t80->["\escale{$X_4$}"]t81->["\escale{$X_5$}"]t82->["\escale{$X_4$}"]t83;
		
			};
\end{scope}
\end{tikzpicture}
\end{center}

\begin{center}
\begin{tikzpicture}[new set = import nodes]
\begin{scope}[nodes={set=import nodes}]
		
	\foreach \j / \i in {0/0, 0/1,   1/0,1/1} { \coordinate (f\j\i) at (\i*3cm, -\j*3cm); }
	
	\foreach \j / \i in {0/0, 0/1,   1/0,1/1}  { \node (f\j\i) at (f\j\i) {\nscale{$f_{\j, \i}$}}; }
	\graph {
	(import nodes);
			f01<-[dashed,  "\escale{$X_0$}"]f00;
			f11<-["\escale{$X_1$}"]f10;
			\foreach \j in {0,1}  {
			f\j0->[bend left=20, , "\escale{$a_0$}"]f\j1;
			f\j0->[bend left=45, , "\escale{$a_1$}"]f\j1;
			f\j0->[bend left=90, , "\escale{$a_2$}"]f\j1;
			f\j0->[swap, bend right =35,  "\escale{$a_3$}"]f\j1;
			f\j0->[swap, bend right =70,  "\escale{$a_4$}"]f\j1;
			}
			
			};
\end{scope}
\begin{scope}[xshift=4cm, nodes={set=import nodes}]%
		
	\foreach \j / \i in {2/0,2/1, 	3/0,3/1} { \coordinate (f\j\i) at (\i*3cm, -\j*3cm+6cm); }
	
	\foreach \j / \i in {2/0,2/1, 	 3/0,3/1}  { \node (f\j\i) at (f\j\i) {\nscale{$f_{\j, \i}$}}; }
	\graph {
	(import nodes);
			f21<-[dashed,"\escale{$X_2$}"]f20;
			f31<-[dashed, "\escale{$X_3$}"]f30;
			\foreach \j in {2,3}  {
			f\j0->[bend left=20,  "\escale{$a_0$}"]f\j1;
			f\j0->[bend left=45,  "\escale{$a_1$}"]f\j1;
			f\j0->[bend left=90,  "\escale{$a_2$}"]f\j1;
			f\j0->[swap, bend right =35,  "\escale{$a_3$}"]f\j1;
			f\j0->[swap, bend right =70,  "\escale{$a_4$}"]f\j1;
			}
			
			};
\end{scope}
\begin{scope}[xshift=8cm, nodes={set=import nodes}]
		
	\foreach \j / \i in { 4/0,4/1,   5/0,5/1} { \coordinate (f\j\i) at (\i*3cm, -\j*3cm+12cm); }
	
	\foreach \j / \i in {4/0,4/1,   5/0,5/1}  { \node (f\j\i) at (f\j\i) {\nscale{$f_{\j, \i}$}}; }
	\graph {
	(import nodes);
			f41<-["\escale{$X_4$}"]f40;
			f51<-[dashed, "\escale{$X_5$}"]f50;
			\foreach \j in {4,5}  {
			f\j0->[bend left=20, , "\escale{$a_0$}"]f\j1;
			f\j0->[bend left=45, , "\escale{$a_1$}"]f\j1;
			f\j0->[bend left=90, , "\escale{$a_2$}"]f\j1;
			f\j0->[swap, bend right =35,  "\escale{$a_3$}"]f\j1;
			f\j0->[swap, bend right =70,  "\escale{$a_4$}"]f\j1;
			}
			
			};
\end{scope}
\end{tikzpicture}
\end{center}\vspace*{-5mm}
\caption{For a fixed $j\in \{0,\dots, 3\}$, the gadgets $T_{0,j},\dots, T_{8,j}$ and $F_0,\dots, F_5$ of the LTS $A$ (Section~\ref{sec:edge_removal}) based on Example~\ref{ex:hs}.
Dashed lines correspond to edges that are removed in accordance to the edge removal $B$ defined for the proof of Lemma~\ref{lem:edge_removal_hs_implies_essp_and_ssp} and correspond to the hitting set $\{X_0,X_2,X_3, X_5\}$.}\label{fig:reduction_edge_removal}
\end{figure}
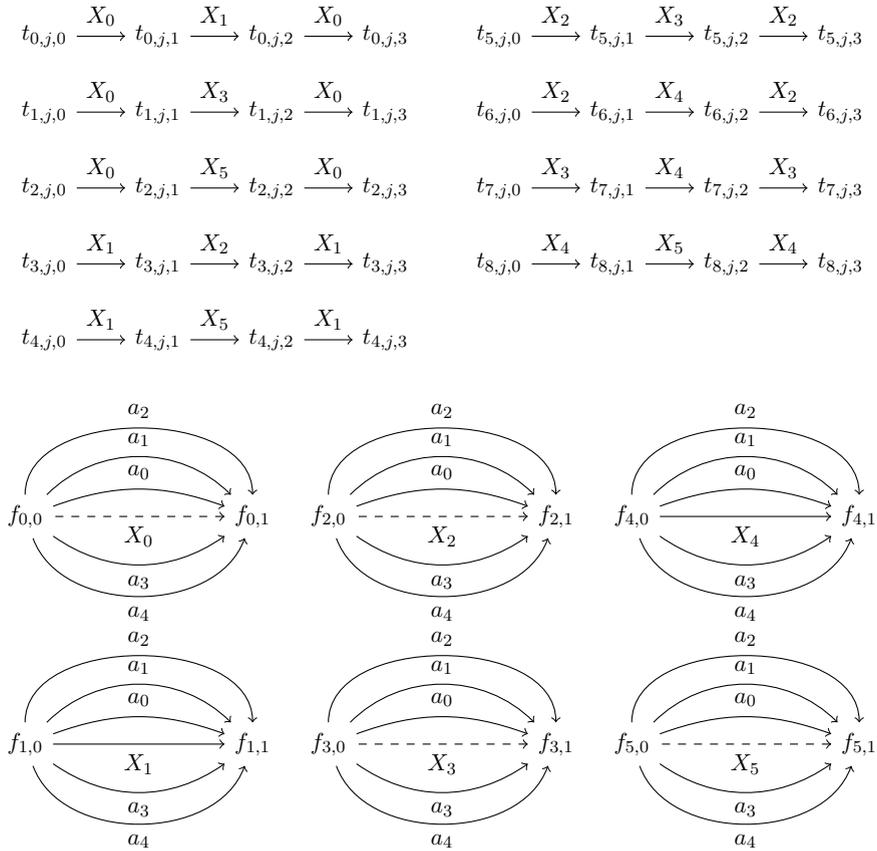

\begin{lemma}\label{lem:edge_removal_essp_implies_hs}
If there is an edge removal $B$ of $A$ by $\K$ that satisfies $\vert \K\vert\leq \kappa$ and has the ESSP, then there is a hitting set of size at most $\lambda=\kappa$ for $(\U,M)$.
\end{lemma}
\begin{proof}
Let $B=(S',E',\delta',\iota)$ be an edge removal of $A$ by $\K$
that satisfies $\vert \K\vert\leq \kappa$ and has the ESSP.
First, we may observe that, without loss of generality, we may assume that all the members of $\K$ are of the form
$f_{i,0}\edge{X_i}f_{i,1}$, so that in particular $B$ remains reachable.

Indeed, for each $i\in \{0,\dots, m-1\}$, since there are more than $\kappa$ copies of $T_{i,0}$, at least one of them is
completely present in $B$ (including the edge from $\iota$). Let us assume it is $T_{i,j}$ for $j\in \{0,\dots, \kappa\}$.
Then we may remove from $\K$ all the members of the form $t_{i,k,h}\edge{X_\ell}t_{i,k,h+1}$ or $\iota\Edge{u_i^k}t_{i,k,0}$ and reinsert them in $T_{i,k}$. This reduces the size of $\K$ as well as the set of ESSAs and,
by choosing $con(u_i^k)=con(u_i^j)$, and $pro(u_i^k)=pro(u_i^j)$,
we get regions that solve the same ESSAs in $T_{i,k}$ as in $T_{i,j}$.

Similarly, since the set $\{a_0,\dots, a_\kappa\}$ contains $\kappa+1$ elements and $\vert \K\vert\leq\kappa$, there is an index $h\in \{0,\dots, \kappa\}$ such that all $a_h$-labeled edges and thus particularly the edge $f_{i,0}\edge{a_h}f_{i,1}$ is present in $B$ for all $i\in \{0,\dots, n-1\}$.
Hence, if we reinsert in $B$ all the missing $a_\ell$-labeled edges, $\ell\in \{0,\dots, \kappa\}$, as well as the missing $v_i$-labeled ones, each region of the old $B$ extends immediately in a region of the new one, solving all the needed remaining ESSAs.

In the following, we show that this implies that the set $Z=\{X_i\in {\U} \mid f_{i,0}\edge{X_i}f_{i,1} \in \K\}$ defines a hitting set, with at most $\lambda=\kappa$ elements, for $(\U,M)$.

Since $B$ has the ESSP,  for each $i\in \{0,\dots, m-1\}$ and $j\in \{0,\dots, \kappa\}$,
there is a region that solves $\alpha=(X_{i_0}, t_{i,j,1})$.
Let $R=(sup, con, pro)$ be such a region.
We argue that there is some $\ell\in \{0,\dots, m_i-1\}$ such that the $X_{i_\ell}$-labeled edge
$f_{i_\ell,0}\Edge{X_{i_\ell}}f_{i_\ell,1}\in \K$, i.e., is not present in $B$.

Assume, for a contradiction, that $f_{i_\ell,0}\Edge{X_{i_\ell}}f_{i_\ell,1}\in B$ for all $\ell\in \{0,\dots, m_i-1\}$.
By $t_{i,j,0}\Edge{X_{i_0}}$, we have $con(X_{i_0}) \leq sup(t_{i,j,0})$;
since $R$ solves $\alpha$, we have that $con(X_{i_0}) > sup(t_{i,j,1})$.
This implies $con(X_{i_0}) > pro(X_{i_0})$, and thus $sup(f_{i_0,0}) > sup(f_{i_0,1})$.
Moreover, by $t_{i,j,m_i}\Edge{X_{i_0}}$, we have that $con(X_{i_0}) \leq sup(t_{i,j,m_i})$, which, by $con(X_{i_0}) > sup(t_{i,j,1})$, implies
\[sup(t_{i,j,1}) <  sup(t_{i,j,m_i}) = sup(t_{i,j,1}) +\sum_{\ell=1}^{m_i-1}\eff(X_{i_\ell})\]
 and thus $\sum_{\ell=1}^{m_i-1}\eff(X_{i_\ell})> 0$.
In particular, there is an $\ell\in \{1,\dots, m_i-1\}$ (so that $i_\ell\neq i_0$)
such that $con(X_{i_\ell}) < pro(X_{i_\ell})$, which implies $sup(f_{i_\ell,0}) < sup(f_{i_\ell,1})$.

Since the edges $f_{i_0,0}\edge{a_0}f_{i_0,1}$ and $f_{i_\ell,0}\edge{a_0}f_{i_\ell,1}$ are both present in $B$,
 by $sup(f_{i_0,0})>sup(f_{i_0,1})$, we obtain $con(a_0) > pro(a_0)$ and, by $sup(f_{i_\ell,0}) < sup(f_{i_\ell,1})$, we get $con(a_0) < pro(a_0)$, which is a contradiction.
Consequently, $\{f_{i_0,0}\Edge{X_{i_0}}f_{i_0,1}, f_{i_\ell,0}\Edge{X_{i_\ell}}f_{i_\ell,1} \}\cap \K\not=\emptyset$ and thus $M_i\cap Z\not=\emptyset$.
Since $i$ was arbitrary, we have $M_i\cap Z\not=\emptyset$ for all $i\in \{0,\dots, m-1\}$ and
 $Z$ defines a hitting set for $(\U,M)$.
\end{proof}

Conversely, we show:

\begin{lemma}\label{lem:edge_removal_hs_implies_essp_and_ssp}
If there is a hitting set with at most $\lambda$ elements for $(\U,M)$, then there is an edge removal $B$ of $A$ with $\vert \K\vert \leq \lambda$ that has the ESSP and the SSP.
\end{lemma}
\begin{proof}
Let $Z=\{X_{j_0},\dots, X_{j_{\lambda-1}}\}$, where $j_0,\dots, j_{\lambda-1}\in \{0,\dots, n-1\}$, be a hitting set with $\lambda$ elements for $(\U,M)$.
(Notice that every hitting set with at most $\lambda$ elements for $(\U,M)$ can be easily extended to a hitting set with exactly $\lambda$ elements.)
Moreover, let $B=(S,E,\delta', \iota)$ be the \LTS\  that originates from $A$ by removing, for all $\ell\in \{0,\dots, \lambda-1\}$, the edges $f_{j_\ell,0}\edge{X_{j_\ell}}f_{j_\ell,1}$, and nothing else. 
Obviously, $B$ is an edge removal of $A$ such that $\vert \K\vert \leq\lambda=\kappa$.
In the following, we argue that $B$ has the SSP and the ESSP.
Let $I=\bigcup_{i=0}^{m-1}(\bigcup_{j=0}^\kappa \{t_{i,j,0}\}) \cup \bigcup_{i=0}^{n-1}\{f_{i,0}\}$ be the set of the initial states of the gadgets of $B$, and let $E'=\U\cup \{a_0,\dots, a_\kappa\}$ be the set of the events of them.

For a start, it is easy to see that if $s$ and $s'$ are states of different gadgets (or if one of it equals $\iota$), then the SSA $(s,s')$ is solvable.
Likewise, one easily finds out that if $e$ is an event and $s$ is a state that do not belong to the same gadget of $B$,
then the ESSA $(e,s)$ is solvable.
Hence, in the following, we restrict our attention to the solvability of separation atoms whose components belong to the same gadgets.
Notice that we define the following regions implicitly according to Remark~\ref{rem:implicit}.
For an even more compact representation, until explicitly stated otherwise, we assume $\iota$ to be mapped to $0$, and we define additionally only the support of the initial states of the gadgets, and $con$, and $pro$ for the events of the gadgets.
According to Remark~\ref{rem:implicit}, one computes the support values of the remaining states of the gadgets.
Moreover, since the $u_i^j$'s, and the $v_i$'s are unique, it is easy to see that there $con$ and $pro$ values can always be chosen suitably.

The following region $R_0=(sup_0, con_0, pro_0)$ solves all the remaining SSA and,
for all $i\in \{0,\dots, n-1\}$ and all $e\in \U \cup \{a_0,\dots, a_{\kappa}\}$, the ESSA $(f_{i,1},e)$:
$\;sup_0(\iota)=0$;
for all $i\in \{0,\dots, m-1\}$, $\T_{0,m_i}^{R_0}=\{u_i^0,\dots, u_i^\kappa\}$;
$\T_{0,1}^{R_0}=\{v_0,\dots, v_{n-1}\}$
and $\T_{1,0}=\U\cup \{a_0, \dots, a_{\kappa}\}$.

Hence, it remains to consider the ESSA of the $T_{i,j}$'s.
Let $i\in \{0,\dots, m-1\}$ and $j\in \{ 0,\dots, \kappa \}$ be arbitrary but fixed.

In the following, we investigate the cases $X_{i_0}\in Z$ and $X_{i_0}\not\in Z$ separately.

\medskip\noindent
\textbf{Case $X_{i_0}\not\in Z$:}
Unless explicitly stated otherwise, let $\ell \in \{1,\dots, m_i-1\}$ be the smallest index such that $X_{i_\ell}\in Z$ (which exists since $Z$ is a hitting set.)
Notice that there are $\ell-1$ variable events between the first occurrence of $X_{i_0}$, at the beginning of  $T_{i,j}$,
and $X_{i_\ell}$, and that there are $m_i-(\ell+1)$ variable events between $X_{i_\ell}$ and the second occurrence of $X_{i_0}$ at the end of $T_{i,j}$.

\medskip
The following region $R_1=(sup_1, con_1, pro_1)$ solves $(X_{i_0}, s)$ for all $s\in \{t_{i,j,1},\dots, t_{i,j,\ell}, \}\cup\{t_{i,j,m_i+1}\}$:
for all $s\in I$,
if $s=t_{i,j,0}$, then $sup(t_{i,j,0})=\ell$, otherwise $sup(s)=\ell\cdot \sum_{i=0}^{m-1} m_i$ (that is, $sup(s)$ is just chosen ``big enough'') and $\T_{\ell,\ell-1}^{R_1}=\{X_{i_0}\}$ (this implies $sup(t_{i,j,1})=\ell-1$);
$\T_{1,0}^{R_1}=\U\setminus \{X_{i_\ell}\}\cup \{a_0,\dots, a_\kappa\}$ and  $\T_{0,m_i-1}^{R_1}=\{X_{i_\ell}\}$.

The following region $R_2=(sup_2, con_2, pro_2)$ solves $(X_{i_0}, s)$ for all $s\in \{t_{i,j,\ell+1},\dots, t_{i,j,m_i-1}\}$:
for all $s\in I$,
if $s=t_{i,j,0}$, then $sup(t_{i,j,0})=m_i-\ell-1$, otherwise $sup_2(s)=\ell\cdot \sum_{i=0}^{m-1} m_i$ (again, $sup(s)$ is ``big enough'') and $\T_{m_i-\ell-1, m_i-\ell}^{R_2}=\{X_{i_0}\}$;
$\T_{0,1}^{R_2}= \U\setminus \{X_{i_\ell}\}\cup \{a_0,\dots, a_\kappa\}$ and
$\T_{m_i, 0}^{R_2}= \{X_{i_\ell}\}$.
Notice that $R_2$ implies $sup_2(t_{i,j,\ell+1})=0$ by $m_i-\ell-1 + (\ell-1)=m_i$ and thus $sup_2(t_{i,j,m_i})=m_i-\ell-1$.

This proves the solvability of $X_{i_0}$.
In the following, we argue for the solvability of the other variables events in $T_{i,j}$.
Let $\ell\in \{1,\dots, m-1\}$ be arbitrary but fixed (that is, $X_{i_\ell}$ is not necessarily in $Z$).
Notice that $X_{i_\ell}$ is preceded by exactly $\ell$ variable events and followed by $m_i - \ell $ variable events (since $X_{i_0}$ occurs twice).

The following region $R_3=(sup_3, con_3, pro_3)$ solves $(X_{i_\ell}, s )$ for all $s\in \{t_{i,j,0},\dots, t_{i,j,\ell-1} \}$:
for all $s\in I$,
if $s=t_{i,j,0}$, then $sup(t_{i,j,0})=0$, otherwise $sup_3(s)=\ell\cdot \sum_{i=0}^{m-1} m_i$;
$\T_{\ell, \ell+1}^{R_3}=\{X_{i_\ell}\}$ and $\T_{0,1}^{R_3}=E'\setminus \{X_{i_\ell}\}$.

The following region $R_4=(sup_4, con_4, pro_4)$ solves $(X_{i_\ell}, s )$ for all $s\in \{t_{i,j,\ell+1},\dots, t_{i,j,m_i+1}\}$:
for all $s\in I$,
if $s=t_{i,j,0}$, then $sup(t_{i,j,0})=m_i+1$, otherwise $sup_4(s)=\ell\cdot \sum_{i=0}^{m-1} m_i$;
$\T_{m_i-\ell+1, m_i-\ell}^{R_4}=\{X_{i_\ell}\}$ and $\T_{1,0}^{R_4}=E'\setminus \{X_{i_\ell}\}$.

Since $\ell$ was arbitrary, this proves that all ESSA of $T_{i,j}$ are solvable.

\smallskip\noindent
\textbf{Case $X_{i_0}\in Z$:}
Recall that there are $m_i-1$ events between the two occurrences of $X_{i_0}$ in $T_{i,j}$.

\medskip
The following region $R_5=(sup_5, con_5, pro_5)$ solves $(X_{i_0}, s)$ for all $s\in \{t_{i,j,1},\dots, t_{i,j,m_i-1}\}\cup\{t_{i,j,m_i+1}\}$:
for all $s\in I$, if $s=t_{i,j,0}$, then $sup_5(s)=m_i-1$, otherwise $sup_4(s)=\ell\cdot \sum_{i=0}^{m-1} m_i$;
$\T_{m_i-1, 0}^{R_5}=\{X_{i_0}\}$ and $\T_{0,1}^{R_5}=E'\setminus \{X_{i_0}\}$.

Similar to the former case, one proves, for all $(e,s)$ with $e\in \{X_{i_1},\dots, X_{i_{m_i-1}}\}$ and $s\in S(T_{i,j})$, if $(e,s)$ is an ESSA of $T_{i,j}$, then $(e,s)$ is solvable.

\medskip
Altogether, by the arbitrariness of $i$ and $j$, we conclude that $B$ has the ESSP.\\
This completes the proof.
\end{proof}

\subsection{Edge Removal aiming at embedding}\label{sec:edge_removal_embedding}%

In the conclusion of~\cite{corr/abs-2002-04841}, Schlachter~and~Wimmel mention en passant that their reduction-technique imply that {Edge Removal for embedding} is NP-complete.
The following theorem re-formulates this claim, and is followed by a full (and new) proof.

\begin{theorem}\label{the:edge_removal_embedding}
\textsc{Edge Removal} aiming at embedding is NP-complete.
\end{theorem}

Like in the previous subsection, it is easy to see that Edge Removal aiming at embedding is in NP.

\paragraph{The General Idea of the Reduction.}
To prove hardness, in the following, we reduce $(\U,M,\lambda)$ polynomially to a pair $(A,\kappa)$ such that there is hitting set of size at most $\lambda$ for $(\U, M)$ if the removal of at most $\lambda$ edges of $A$ yields an \LTS\  $B$ that has the SSP.
Conversely, if $(\U,M,\lambda)$ has a fitting hitting set, then there are at most $\kappa$ edges of $A$ whose removal yields an \LTS\  that has the SSP.

\paragraph{The Reduction.}
Again, we shall have $\kappa=\lambda$.

For every $i\in \{0,\dots, m-1\}$, and for every $j\in \{0,\dots, \kappa\}$, the \LTS\  $A$ has the following gadget $T_{i,j}$, that uses the event $k_i$, and the elements of $M_i=\{ X_{i_0},\dots, X_{m_i-1}\}$ as events:\vspace*{-1mm}

\begin{center}
\begin{tikzpicture}[new set = import nodes]
\begin{scope}[nodes={set=import nodes}]
	
	\node (t) at (-1,0){$T_{i,j}=$};
	\foreach \i in {0,...,2} {\coordinate (\i) at (\i*2cm,0);}
	\foreach \i in {3} {\coordinate (\i) at (\i*2cm+6,0);}
	\foreach \i in {4} {\coordinate (\i) at (8.75,0);}
	\foreach \i in {0,1} {\node (\i) at (\i) {\nscale{$t_{i,j,\i}$}};}
	\node (2) at (2) {$\dots$};
	\node (3) at (3) {\nscale{$t_{i,j,m_i}$}};
	\node (4) at (4) {\nscale{$t_{i,j,m_i+1}$}};
	\graph {
	(import nodes);
			0->["\escale{$k_i$}"]1->["\escale{$X_{i_0}$}"]2->["\escale{$X_{i_{m_i-2}}$}"]3->["\escale{$X_{i_{m_i-1}}$}"]4;
			
			};
\end{scope}
\end{tikzpicture}\vspace*{-1mm}
\end{center}

Furthermore, for every $i\in \{0,\dots, m-1\}$, and for every $j\in \{0,\dots, \kappa\}$, the \LTS\  $A$ has the following directed cycle $D_{i,j}$, where the event $a$ occurs $m_i$ times in a row followed by $k_i$, which closes the cycle:
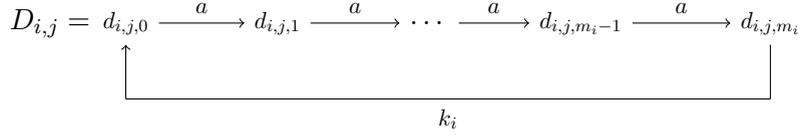
\vspace*{-1mm}

\begin{center}
\begin{tikzpicture}[new set = import nodes]
\begin{scope}[nodes={set=import nodes}]
	
	\node (d) at (-1,0){$D_{i,j}=$};
	\node (d0) at (0,0) {\nscale{$d_{i,j,0}$}};
	\node (d1) at (2,0) {\nscale{$d_{i,j,1}$}};
	\node (dots) at (4,0) {$\dots$};
	\node (d2) at (6,0) {\nscale{$d_{i,j,m_i-1}$}};
	\node (d3) at (8.5,0) {\nscale{$d_{i,j,m_i}$}};
	\coordinate (d40) at (8.5,-1);
	\coordinate (d41) at (0,-1);
	\graph {
	(import nodes);
			d0->["\escale{$a$}"]d1->["\escale{$a$}"]dots->["\escale{$a$}"]d2->["\escale{$a$}"]d3;
			d3--d40--["\escale{$k_i$}"]d41->d0;
			
			};
\end{scope}
\end{tikzpicture}\vspace*{-1mm}
\end{center}

For all $i\in \{0,\dots, n-1\}$, the \LTS\  $A$ also has following gadget $F_i$ that synchronizes the event $a$ with the variable
event $X_i$:
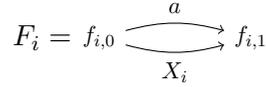
\vspace*{-2mm}

\begin{center}
\begin{tikzpicture}[new set = import nodes]
\begin{scope}[nodes={set=import nodes}]
	
	\node (f) at (-0.75,0){$F_i=$};
	\node (f0) at (0,0) {\nscale{$f_{i,0}$}};
	\node (f1) at (2,0) {\nscale{$f_{i,1}$}};
	
	\graph {
	(import nodes);
			f0->[bend left = 15,"\escale{$a$}"]f1;
			f0->[bend right = 15,swap, "\escale{$X_i$}"]f1;
			};
\end{scope}
\end{tikzpicture}\vspace*{-2mm}
\end{center}

Finally, in order to connect the introduced gadgets, and to ensure reachability, the \LTS\  $A$ has the initial state $\iota$, and the following
edges:
\begin{itemize}
\itemsep=0.9pt
\item
for all $i\in \{0,\dots, m-1\}$, and all $j\in \{0,\dots, \kappa\}$, the \LTS\  $A$ has the edges $\iota\Edge{u_i^j}t_{i,j,0}$, and $\iota\Edge{v_i^j}d_{i,j,0}$.
\item
For all $i\in \{0,\dots, n-1\}$, the \LTS\  $A$ has the edge $\iota\edge{w_i}f_{i,0}$.
\end{itemize}

The result is the reachable \LTS\  $A=(S,E,\delta, \iota)$.

\begin{lemma}\label{lem:edge_removal_ssp_implies_hs}
If there is an edge removal $B$ of $A$ that has the SSP and satisfies $\vert \K\vert\leq \kappa$, then there is a hitting set with at most $\lambda=\kappa$ elements for $(\U,M)$.
\end{lemma}
\begin{proof}
Let $B$ be an edge removal of $A$ that has the SSP, and satisfies $\vert \K\vert\leq \kappa$.
Like in the proof of Lemma~ \ref{lem:edge_removal_ssp_implies_hs}, without loss of generality, we may assume that no edge in $\K$ arises from $\iota$, nor from some $T_{i,j}$, nor from some $D_{i,j}$, and if one arises from some $F_i$ the other edge from $F_i$ is not in $\K$, so that in particular $B$ remains reachable.
If some edge from $\iota$ has been removed, we may reintroduce it and choose its effect in order to keep the SSA we had before.
Since there are more than $\kappa$ copies $T_{i,j}$, at least one is completely in $B$;
then we may complete the other copies in order to keep their SSAs.
The same is true for $D_{i,j}$.
Now, if both $a$ and $X_i$ are missing in some $F_i$, we may reintroduce $a$;
indeed, when we separate $d_{i,j,0}$ from $d_{i,j,1}$, the effect of $a$ is non-null, then $f_{i,0}$ will be separated from $f_{i,1}$ (it may be necessary to choose $pro(w_i)$ high enough in order to allow $a$ in $f_{i,0}$).

\medskip
In the following, we argue that
\[
Z=\{X_i \mid i\in \{0,\dots, n-1\} \text{ and } \{f_{i,0}\edge{X_i}f_{i,1}, f_{i,0}\edge{a}f_{i,1}\}\cap \K\not=\emptyset\}
\]
is a hitting set for $(\U,M)$.

\medskip
Let $i\in \{0,\dots, m-1\}$ be arbitrary but fixed. From the hypothesis on $B$ above, $T_{i,0}$, and $D_{i,0}$ are completely present in $B$.
By the SSP of $B$, there is a region that solves the SSA $\alpha=(t_{i,0,0}, t_{i,0,m_i+1})$.
Let $R=(sup, sig)$ be such a region.
We argue that $Z\cap M_i\not=\emptyset$:
Assume, for a contradiction, that the opposite is true, that is, $Z\cap M_i=\emptyset$.
Then, for all $\ell\in \{0,\dots, m_i-1\}$, the gadget $F_\ell$ is completely present in $B$.
By Lemma~\ref{lem:effect_on_a_path}, we have that $sup(d_{i,0,m_i})=sup(d_{i,0,0})+m_i\cdot \eff(a)$, and $sup(d_{i,0,0})=sup(d_{i,0,m_i})+ \eff(k_i)$, which implies $\eff(k_i)=-m_i\cdot \eff(a)$.
Moreover, since the gadget $F_\ell$ is present in $B$, we have that $\eff(X_{i_\ell})=\eff(a)$ for all $\ell\in \{0,\dots, m_i-1\}$.
Finally, again by Lemma~\ref{lem:effect_on_a_path}, we have\vspace*{-2mm}
\begin{align*}
sup(t_{i,0,m_i+1}) & =sup(t_{i,0,0})+\eff(k_i)+\sum_{\ell=0}^{m_i-1}\eff(X_{i_\ell})\\
			&= sup(t_{i,0,0})-m_i\cdot\eff(a) + m_i\cdot\eff(a)\\
			&= sup(t_{i,0,0}),
\end{align*}
which contradicts that $R$ solves $\alpha$.
Hence, there is an $\ell\in \{0,\dots, m_i-1\}$, such that $f_{i_\ell,0}\Edge{X_{i_\ell}}f_{i_\ell,1}\in \K$ or $f_{i_\ell,0}\edge{a}f_{i_\ell,1}\in \K$, which implies $Z\cap M_i\not=\emptyset$.
Since $i$ was arbitrary, this is simultaneously true for all sets of $M$ and $Z$ is a hitting set for $(\U,M)$, which completes the proof.
\end{proof}

Conversely, we show:

\begin{lemma}\label{lem:edge_removal_hs_implies_ssp}
If there is a hitting set with at most $\lambda$ elements for $(\U,M)$, then there is an edge removal $B$ of $A$ that has the SSP, and satisfies $\vert \K\vert \leq \lambda$.
\end{lemma}
\begin{proof}
Let $Z=\{X_{\ell_0},\dots, X_{\ell_{\lambda-1}}\}$, where $\ell_0,\dots, \ell_{\lambda-1}\in \{0,\dots, n-1\}$, be a hitting set with $\lambda$ elements for $(\U, M)$ (which exists whenever there is a hitting set with at most $\lambda$ elements).
Moreover, let $B=(S,E,\delta',\iota)$ be the \LTS\  that originates from $A$ by removing, for all $j\in \{0,\dots, \lambda-1\}$, the edge $f_{\ell_j,0}\Edge{X_{\ell_j}}f_{\ell_j,1}$.
One easily verifies that $B$ is a well-defined edge removal of $A$ that satisfies $\vert \K\vert\leq \lambda=\kappa$,  and that $B$ is reachable.
The definition of $B$ implies that $B$ and $A$ differ (only) with respect to the gadgets $F_{\ell_0},\dots, F_{\ell_{\lambda-1}}$.
However, by a little abuse of notation, and for the sake of readability, in the following, we will refer to $F_0,\dots, F_{n-1}$ also as the gadgets of $B$, where we always keep in mind that the edge $f_{\ell_j,0}\edge{X_{\ell_j}}f_{\ell_j,1}$ is not present in $F_{\ell_j}$ for all $j\in \{0,\dots, \lambda-1\}$.

\medskip
First of all, it is easy to see, that if $G$ and $G'$ are two distinct gadgets of $A$, then,
for all states $s\in S(G)\cup\{\iota\}$ and $s'\in S(G')$, the SSA $(s,s')$ is solvable.
Hence, it remains to argue that every SSA $(s,s')$ of $B$ is solvable if $s$ and $s'$ belong to the same gadget of $B$.

The following region $R=(sup, con, pro)$ solves all the corresponding SSAs in one blow:
For all $i\in \{0,\dots, m-1\}$, and all $j\in \{0,\dots, n-1\}$, let $sup(d_{i,j,0})=0$, and $sup(t_{i,j,0})=m_i$;
for all $i\in \{0,\dots, n-1\}$, let $sup(f_{i,0})=0$;
for all $e\in \{a\}\cup  (\U\setminus Z)$, let $(con(a), pro(a))=(0,1)$;
for all $i\in \{0,\dots, m-1\}$, let $(con(k_i), pro(k_i))=(m_i,0)$;
for all $j\in \{0,\dots, \lambda-1\}$, let $(con(X_{\ell_j}), pro(X_{\ell_j}))=(0,\sum_{i=0}^{m-1} m_i)$.
Obviously, $R$ solves all SSA of the $D_{i,j}$'s, and the $F_i$'s.

Let $i\in \{0,\dots, m-1\}$, and $j\in \{0,\dots, \kappa\}$ be arbitrary but fixed and let $h\in \{0,\dots, m_i-1\}$ be the smallest index such that $X_{i_h}$ belongs to the hitting set $Z$, that is, $X_{i_h}\in Z$.
Then the definition of $R$ ensures that $sup(d_{i,j,1})< \dots < sup(d_{i,j,i_h+1}) < sup(d_{i,j,0}) < sup(d_{i,j,i_h+2})<  \dots <sup(d_{i,j,m_i+1}) $.
By the arbitrariness of $i$, and $j$, this completes the proof.
\end{proof}

\section{The complexity of Event Removal}\label{sec:event_removal}%

In this section, we are looking for implementable event removals that remove only a bounded number of events.
By the connection between implementations and separation properties stated by Lemma~\ref{lem:badouel}, the following problems arise:

\noindent
\fbox{\begin{minipage}[t][1.8\height][c]{0.97\textwidth}
\begin{decisionproblem}
  \problemtitle{\textsc{Event Removal for Embedding}}
  \probleminput{An \LTS\  $A=(S,E,\delta,\iota)$, a natural number $\kappa$.}
  \problemquestion{Does there exist an event removal $B$ for $A$ by ${\E}$ that has the SSP and satisfies $\vert{\E}\vert\leq \kappa$?}
\end{decisionproblem}
\end{minipage}}
\smallskip

\noindent
\fbox{\begin{minipage}[t][1.8\height][c]{0.97\textwidth}
\begin{decisionproblem}
  \problemtitle{\textsc{Event Removal for Language-Simulation}}
  \probleminput{An \LTS\  $A=(S,E,\delta,\iota)$, a natural number $\kappa$.}
  \problemquestion{Does there exist an event removal $B$ for $A$ by ${\E}$ that has the ESSP and satisfies $\vert{\E}\vert\leq \kappa$?}
\end{decisionproblem}
\end{minipage}}
\smallskip

\noindent
\fbox{\begin{minipage}[t][1.8\height][c]{0.97\textwidth}
\begin{decisionproblem}
  \problemtitle{\textsc{Event Removal for Realization}}
  \probleminput{An \LTS\  $A=(S,E,\delta,\iota)$, a natural number $\kappa$.}
  \problemquestion{Does there exist an event removal $B$ for $A$ by ${\E}$ that has the ESSP and the SSP and satisfies $\vert{\E}\vert\leq \kappa$?}
\end{decisionproblem}
\end{minipage}}
\medskip

The following theorem states the main result of this section:
\begin{theorem}\label{the:event_removal}
{Event Removal for embedding}, {Event Removal for Language Simulation}, and {Event Removal for realization} are NP-complete.
\end{theorem}

Obviously, the problems addressed by Theorem~\ref{the:event_removal} belong to NP:
If there is a fitting event removal $B$ for $A$, then a Turing machine $T$ can guess ${\E}$ non-deterministically in time polynomial in the size of $A$, since $\vert {\E}\vert \leq \vert E\vert$.
After that, $T$ can deterministically compute $B$ and a witness for the property in question in polynomial time by Lemma~\ref{lem:badouel}, since the size of $B$ is bounded by the size of $A$.

Hence, to complete the proof of the NP-completeness part of Theorem~\ref{the:event_removal}, it remains to prove the hardness-part.
In order to do that, we provide a reduction of  \textsc{Hitting Set} that reduces the input $(\U, M,\lambda)$ to an instance $(A,\kappa)$ with \LTS\ $A$ and natural number $\kappa$ such that is a hitting set of size at most $\lambda$ for $(\U,M)$ if and only if $A$ can be made implementable by removing at most $\kappa$ events.

\medskip
For a start, we first define $\kappa=\lambda$.
Moreover, for every $i\in \{0,\dots, m-1\}$, the \LTS\  $A$ has the following directed path on which the elements of $M_i=\{X_{i_0},\dots, X_{m_i-1}\}$ occur as event, which are encompassed by the event $k_i$:
\begin{center}
\begin{tikzpicture}[new set = import nodes]
\begin{scope}[nodes={set=import nodes}]
	
	\node (t) at (-0.75,0){$T_i=$};
	\foreach \i in {0,...,2} {\coordinate (\i) at (\i*2cm,0);}
	\foreach \i in {3} {\coordinate (\i) at (\i*2cm+6,0);}
	\foreach \i in {4} {\coordinate (\i) at (8.5,0);}
	\foreach \i in {0,1} {\node (\i) at (\i) {\nscale{$t_{i,\i}$}};}
	\node (2) at (2) {$\dots$};
	\node (3) at (3) {$t_{i,m_i+1}$};
	\node (4) at (4) {$t_{i,m_i+2}$};
	\graph {
	(import nodes);
			0->["\escale{$k_i$}"]1->["\escale{$X_{i_0}$}"]2->["\escale{$X_{i_{m_i-1}}$}"]3->["\escale{$k_i$}"]4;
			};
\end{scope}
\end{tikzpicture}
\end{center}

Moreover, for every $i\in \{0,\dots, m-1\}$, the \LTS\  $A$ has the following directed cycle at which the elements of $M_i=\{X_{i_0},\dots, X_{m_i-1}\}$ occur consecutively as events:

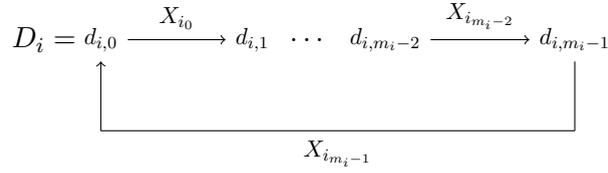
\eject
\hbox{}
\vspace*{-11mm}
\begin{center}
\begin{tikzpicture}[new set = import nodes]
\begin{scope}[nodes={set=import nodes}]
	
	\node (d) at (-0.75,0){$D_i=$};
	\node (d0) at (0,0) {\nscale{$d_{i,0}$}};
	\node (d1) at (2,0) {\nscale{$d_{i,1}$}};
	\node (dots) at (2.75,0) {$\dots$};
	\node (d2) at (3.75,0) {\nscale{$d_{i,m_i-2}$}};
	\node (d3) at (6.25,0) {\nscale{$d_{i,m_i-1}$}};
	\coordinate (d40) at (6.25,-1.2);
	\coordinate (d41) at (0,-1.2);
	\graph {
	(import nodes);
			d0->["\escale{$X_{i_0}$}"]d1;
			d2->["\escale{$X_{i_{m_i-2}}$}"]d3;
			d3--d40--["\escale{$X_{i_{m_i-1}}$}"]d41->d0;
			
			};
\end{scope}
\end{tikzpicture}\vspace*{-2mm}
\end{center}

In order to connect the gadgets, and to ensure that every (reasonable) event removal of $A$ is a reachable  LTS,
that is, every state is reachable from the initial state $\iota$ by a directed  path, we add the following edges:
\begin{itemize}
\itemsep=0.85pt
\item
for every $i\in \{0,\dots, m-1\}$, and every $j\in \{0,\dots, m_i+2\}$, the \LTS\  $A$ has the edge $\iota\Edge{u_i^j}t_{i,j}$;\vspace*{-1mm}
\item
for every $i\in \{0,\dots, m-1\}$, and every $j\in \{0,\dots, m_i-1\}$, the \LTS\  $A$ has the edge $\iota\Edge{v_i^j}d_{i,j}$.
\end{itemize}

The result is the \LTS\  $A=(S,E,\delta, \iota)$.
In the following, we prove the functionality of $A$.
(Recall that ${\E}$ refers to the events removed from $A$.)

\begin{lemma}\label{lem:event_removal_essp_or_ssp_implies_hs}
If there is an event removal $B=(S,E',\delta',\iota)$ of $A$  by ${\E}$ that has the ESSP or the SSP, and satisfies $\vert \E\vert \leq \kappa$, then there is a hitting set of size at most $\lambda$ for $(\U,M)$.
\end{lemma}
\begin{proof}
In the following, we argue that the set $Z=\U\cap\E$ defines a hitting set with at most $\lambda$ elements for $(\U,M)$.
In order to do that, we show that, for every $i\in \{0,\dots, m-1\}$, there has to be at least one event $X\in M_i$ that does not occur in $B$, and thus is removed from $A$.
Note first that we may assume that no $u_i^j$ or $v_i^j$ belongs to ${\E}$, since otherwise it is always possible to reintroduce them in $B$, with adequate $con/pro$, while keeping the regions solving the ESSAs and SSAs.
We may thus assume that $B$ is reachable.

Assume that $B$ has the ESSP (the SSP).
Let $i\in \{0,\dots, m-1\}$ be arbitrary but fixed.
Since $B$ has the ESSP (the SSP), there is a region that solves the ESSA $\alpha=(k_i,t_{i,1})$ (the SSA $\beta=(t_{i,1}, t_{i,m_i+1})$).
Let $R=(sup, con, pro)$ be such a region.
Assume, for a contradiction, that $M_i\cap\E=\emptyset$.
By Lemma~\ref{lem:effect_on_a_path}, we have that $sup(d_{i,0})=sup(d_{i,0})+\sum_{j=0}^{m_{i}-1}\eff(X_{i_j})$, which implies $\sum_{j=0}^{m_{i}-1}\eff(X_{i_j})=0$.
Since $R$ solves $\alpha$ (solves $\beta$), and $k_i$ occurs at $t_{i,m_i+1}$, we have that $con(k_i)>sup(t_{i,1})$, and $con(k_i)\leq sup(t_{i,m_i+1})$ (we have $sup(t_{i,1})\not=sup(t_{i,m_i+1})$).
Hence, we have that $sup(t_{i,m_i+1})=sup(t_{i,1})+ \sum_{j=0}^{m_{i}-1}\eff(X_{i_j})  \not= sup(t_{i,1})$, and thus $\sum_{j=0}^{m_{i}-1}\eff(X_{i_j}) \not= 0$, which is a contradiction.
Hence, we have that $M_i\cap\E\not=\emptyset$.
Since $i$ was arbitrary, and $\vert Z\vert = \vert \U\cap\E\vert \leq \vert \E\vert \leq \kappa=\lambda $, this implies that $Z$ is a hitting set of size at most $\lambda$ for $(\U,M)$.
The claim follows.
\end{proof}

\begin{lemma}\label{lem:event_removal_hs_implies_essp_and_ssp}
If there is a hitting set of size at most $\lambda$ for $(\U,M)$, then there is an event removal $B=(S,E',\delta',\iota)$ of $A$
 by ${\E}$ that has the ESSP and the SSP, and satisfies $\vert \E\vert \leq \kappa$.
\end{lemma}
\begin{proof}
Let $U=\{u_i^j \mid i\in \{0,\dots, m-1\}, j\in \{0,\dotsm m_i+2\}\}$, and $V=\{v_i^j \mid i\in \{0,\dots, m-1\}, j\in \{0,\dotsm m_i-1\}\}$.
Let $Z$ be a hitting set of size $\lambda$ for $(\U,M)$.
Let $B$ be the \LTS\  that results from $A$ by removing exactly the events of $Z$ (and the corresponding edges), and nothing else.
Obviously, $B$ is a suitable event removal of $A$ that satisfies $\vert \E\vert \leq\lambda= \kappa$, and is reachable.

For a start, the following region $R_1=(sup_1,con_1,pro_1)$ solves $(s,e)$, and $(\iota,s)$ for all $s\in S\setminus \{\iota\}$, and all $e\in U\cup V$:
$sup_1(\iota)=1$, $con_1(e)=1$ for each $e\in U\cup V$, all the other images are $0$.

The following observation is crucial for our further argumentation:
Since $Z$ is a hitting set, and thus intersects with the event set of $T_i$, and $D_i$ for all $i\in \{0,\dots, m-1\}$, we observe that $B\setminus \{s\edge{e}s'\mid e\in U\cup V\}$ consists of simple directed paths, say
$P=s_0^0\Edge{e_1^0}\dots \Edge{e_{n_0}^0}s_{n_0}^0,\dots, P_y=s_0^y\Edge{e_1^y}\dots \Edge{e_{n_y}^y}s_{n_y}^y$ for some $y\in \N$,
on which every event of $B$ occurs at most once.

The following region $R_2=(sup_2, con_2, pro_2)$ solves all the remaining SSA of $B$:
let $sup(\iota)=0$, and,
for all $i\in \{0,\dots, y\}$, and all $j\in \{0,\dots, n_i\}$, let $sup(s^i_j)=\sum_{\ell=0}^{i-1}n_\ell + i + j$, and,
for all $e\in (\U\setminus Z)\cup\{k_0,\dots, k_{m-1}\}$, let $con(e)=0$, and $pro(e)=1$, and,
for all $e\in U\cup V$, if $\iota\edge{e}s$ is the unique e-labeled edge in $B$, then $con(e)=0$, and $pro(e)=sup(s)$.
(That is, we successively increment the support values of the states of $B$ by one, and the supports of the $s_0^i$'s are chosen in order to get different values for all the simple paths.)

\medskip
Let $a\in (\U\setminus Z)\cup\{k_0,\dots, k_{m-1}\}$ be arbitrary but fixed,
and let $P_{i_0},\dots, P_{i_\ell}$ be the paths in which $a$ occurs.
The following region $R_4=(sup_4,con_4,pro_4)$ solves $(a,s)$ for all states $s\in S\setminus (\bigcup_{j=0}^\ell S(P_{i_j}))$:
for all $s\in S$, if $s\in \bigcup_{j=0}^\ell S(P_{i_j})$, then $sup(s)=1$, otherwise $sup(s)=0$;
for all $e\in (\U\setminus Z)\cup\{k_0,\dots, k_{m-1}\}$, if $e=a$, then $(con(e),pro(e))=(1,1)$, otherwise $(con(e),pro(e))=(0,0)$;
for all $e\in U\cup V$, if $\iota\edge{e}s$ is the unique e-labeled edge in $B$, then $(con(e),pro(e))=(0,sup(s))$.

\medskip
It remains to argue that the events of $(\U\setminus Z)\cup\{k_0,\dots, k_{m-1}\}$ can be separated within the paths at which they occur.
This can be seen as follows:
Let $i\in \{0,\dots, y\}$, and $j\in \{1,\dots, n_i\}$ be arbitrary but fixed.
We observe that there are mappings $R_5=(sup_5,con_5,pro_5)$ such that $sup_5(s_{j-1}^i)=1$, and $sup_5(s)=0$ for all $s\in S(P_i)\setminus \{s_{j-1}^i\}$, and $(con_5(e_j^i), pro_5(e_j^i))=(1,0)$, and $(con_5(e_{j-1}^i), pro_5(e_{j-1}^i))=(0,1)$ if $j\geq 2$, and $(con_5(e), pro_5(e))=(0,0)$ for all $e\in E(P_i)\setminus \{e_{j-1}^i, e_j^i\}$.
Moreover, since every path of $P_0,\dots, P_{i-1},P_{i+1},\dots,P_y$ contains every event of $B$ at most once, and since the events of $U\cup V$ occur exactly once in $B$, it is easy to see that $R_5$ can be extended to a region of $B$, which, for all $s\in S(P_i)\setminus \{s_{j-1}^i\}$, solves $(e_j^i, s)$ by $con(e_j^i)>s$.
By the arbitrariness of $i$, this completes the proof.
\end{proof}

\section{The complexity of State Removal}\label{sec:state_removal}%

In this section, we are interested in finding implementable state removals of $A$ that remove only a restricted number of states.
Again justified by Lemma~\ref{lem:badouel}, this task corresponds to the following decision problems:\smallskip

\noindent
\fbox{\begin{minipage}[t][1.8\height][c]{0.97\textwidth}
\begin{decisionproblem}
  \problemtitle{\textsc{State Removal for Embedding}}
  \probleminput{An \LTS\  $A=(S,E,\delta,\iota)$, a natural number $\kappa$.}
  \problemquestion{Does there exist a state removal $B$ for $A$ by ${\mS}$ that has the SSP and satisfies $\vert{\mS}\vert\leq \kappa$?}
\end{decisionproblem}
\end{minipage}}
\smallskip

\noindent
\fbox{\begin{minipage}[t][1.8\height][c]{0.97\textwidth}
\begin{decisionproblem}
  \problemtitle{\textsc{State Removal for Language-Simulation}}
  \probleminput{An \LTS\  $A=(S,E,\delta,\iota)$, a natural number $\kappa$.}
  \problemquestion{Does there exist a state removal $B$ for $A$ by ${\mS}$ that has the ESSP and satisfies $\vert{\mS}\vert\leq \kappa$?}
\end{decisionproblem}
\end{minipage}}
\smallskip

\noindent
\fbox{\begin{minipage}[t][1.8\height][c]{0.97\textwidth}
\begin{decisionproblem}
  \problemtitle{\textsc{State Removal for Realization}}
  \probleminput{An \LTS\  $A=(S,E,\delta,\iota)$, a natural number $\kappa$.}
  \problemquestion{Does there exist a state removal $B$ for $A$ by ${\mS}$ that has the ESSP and the SSP and satisfies $\vert{\mS}\vert\leq \kappa$?}
\end{decisionproblem}
\end{minipage}}\smallskip
\smallskip

In the following Section~\ref{sec:state_removal_realization}, we show that state removal aiming at embedding or realization is NP-complete.
After that, we show that this is also true if we aim at language simulation in Section~\ref{sec:state_removal_language_simulation}.

\subsection{State Removal aiming at language-simulation or realization}\label{sec:state_removal_realization}%

\begin{theorem}\label{the:state_removal_realization}
{State Removal for Language-Simulation} and  {State Removal for Realization} are NP-complete.
\end{theorem}

First of all, the problems are in NP:
If there is a suitable state removal, then the set ${\mS}$ can be guessed non-deterministically in polynomial time;
after that, $B$ and witnesses for the corresponding separation properties can be computed deterministically in polynomial time by Lemma~\ref{lem:badouel}.
The proof of the hardness-part is based again on a reduction of \textsc{Hitting Set} that transforms the input $(\U, M, \lambda)$ into an instance $(A,\kappa)$.
In particular, we will reuse the reduction that has been used for the proof of Theorem~\ref{the:edge_removal_realization}.
So let $\kappa=\lambda$, and let $A$ be the \LTS\  that has been defined in Section~\ref{sec:edge_removal_realization}.

\begin{lemma}\label{lem:state_removal_essp_implies_hs}
If there is a state removal $B$ of $A$ by $\mS$ that has the ESSP and satisfies $\vert \mS\vert \leq \kappa$, then there is a hitting set of size at most $\lambda$ for $(\U,M)$.
\end{lemma}
\begin{proof}
Let $B$ be a state removal of $A$ by $\mS$  that has the ESSP and satisfies $\vert \mS\vert \leq \kappa$.
Without loss of generality, we may assume that $\mS\subseteq\{f_{j,1}|j\in\{0,\ldots,n-1\}\}$,
so that in particular $B$ is reachable.
Indeed, for each $i\in\{0,\ldots,m-1\}$, there are $\kappa+1$ copies of $T_i$, hence at least one of them is entirely in $B$,
but then we may keep all of them without destroying the ESSP.
For each $j\in\{0,\ldots,n-1\}$, if $f_{j,0}$ and $f_{j,1}$ are both in $\mS$, we may suppress $f_{j,0}$ without destroying ESSP;
if $f_{j,0}$ is in $\mS$ but not $f_{j,1}$, we may replace the first one by the latter in $\mS$, still without destroying ESSP.

\medskip
In the following, we will show that
\[
Z=\{X_j\mid j\in \{0,\dots, n-1\} \text{ and } f_{j,1}\in\mS\}
\]
defines a suitable hitting set for $(\U,M)$, that is, $Z\cap M_i\not=\emptyset$ for each $i\in\{0,\ldots,m-1\}$.

\drop{
Let $i\in \{0,\dots, m-1\}$ be arbitrary but fixed.
By $\vert \mS\vert \leq \kappa$, there is a $j\in \{0,\dots, \kappa\}$ such that $T_{i,j}$ is completely present in $B$.
Since $B$ has the ESSP, there is a region $R=(sup, con pro)$, there is a region that solves the ESSA $(X_{i_0}, t_{i,j,1})$.
Similar to the proof of Lemma~\ref{lem:edge_removal_essp_implies_hs}, one argues that this implies that there has to be an $\ell\in \{1,\dots, m_i-1\}$, such that the edges $f_{i_\ell,0}\edge{X_{i_\ell}}f_{i_\ell,1}$, and the edges $f_{i_\ell,0}\edge{a_h}f_{i_\ell,1}$, where $h\in \{0,\dots, \kappa\}$, can not be simultaneously present in $B$.
Since $B$ originates from $A$ be the removal of states, this implies $\{f_{i_\ell,0}, f_{i_\ell,1}\}\cap \mS\not=\emptyset$.
By the arbitrariness of $i$, and by $\vert \mS\vert \leq \kappa$, we get $Z\cap M_i\not=\emptyset$ for all $i\in \{0,\dots, m-1\}$, and $\vert Z\vert \leq\kappa=\lambda$, respectively.
}

\medskip
Assume, for a contradiction, that for some $i\in\{0,\ldots,m-1\}$ and for all $\ell\in\{0,\ldots,m_i-1\}$, $F_{i_\ell}$ is entirely in $B$.
Since $B$ has the ESSP, there is a region that solves $\alpha=(X_{i_0}, t_{i,0,1})$.
Let $R=(sup, con, pro)$ be such a region.
We argue that there is some $\ell\in \{0,\dots, m_i-1\}$ such that the $X_{i_\ell}$-labeled edge $f_{i_\ell,0}\edge{X_{i_\ell}}f_{i_\ell,1}$ is not present in $B$.

Assume that this is not true and $f_{i_\ell,0}\Edge{X_{i_\ell}}f_{i_\ell,1}\in B$ for all $\ell\in \{0,\dots, m_i-1\}$.
By $t_{i,0,0}\Edge{X_{i_0}}$, we have $con(X_{i_0}) \leq sup(t_{i,0,0})$;
since $R$ solves $\alpha$, we have that $con(X_{i_0}) > sup(t_{i,0,1})$.
This implies $con(X_{i_0}) > pro(X_{i_1})$, and thus $sup(f_{i_0,0}) > sup(f_{i_0,1})$.
Moreover, by $t_{i,0,m_i}\Edge{X_{i_0}}$, we have that $con(X_{i_0}) \leq sup(t_{i,0,m_i})$, which, by $con(X_{i_0}) > sup(t_{i,0,1})$, implies $sup(t_{i,0,1}) <  sup(t_{i,0,m_i}) = sup(t_{i,0,1}) +\sum_{\ell=1}^{m_i-1}\eff(X_{i_\ell})$, and thus $\sum_{\ell=1}^{m_i-1}\eff(X_{i_\ell})> 0$.
In particular, there is an $\ell\in \{1,\dots, m_i-1\}$ (so that $i_\ell\neq i_0$) such that $con(X_{i_\ell}) < pro(X_{i_\ell})$, which implies $sup(f_{i_\ell,0}) < sup(f_{i_\ell,1})$.

Since the edges $f_{i_0,0}\edge{a_0}f_{i_0,1}$ and $f_{i_\ell,0}\edge{a_0}f_{i_\ell,1}$ are both present in $B$, by $sup(f_{i_0,0})>sup(f_{i_0,1})$, we obtain $con(a_0) > pro(a_0)$ and, by $sup(f_{i_\ell,0}) < sup(f_{i_\ell,1})$, we get $con(a_0) < pro(a_0)$, which is a contradiction.
Consequently, either $f_{i_0,1}\in \mS$ or $f_{i_\ell,1}\in\mS$ and thus $M_i\cap Z\not=\emptyset$.
Since $i$ was arbitrary, we have $M_i\cap Z\not=\emptyset$ for all $i\in \{0,\dots, m-1\}$ and $Z$ defines a hitting set for $(\U,M)$.
The claim follows.
\end{proof}

Conversely, the following lemma states that, if $(\U,M,\lambda)$ allows a positive decision, then so does $(A,\kappa)$:

\begin{lemma}\label{lem:state_removal_hs_implies_essp_and_ssp}
If there is a hitting set of size at most $\lambda$ for $(\U,M)$, then there is a state removal $B$ of $A$ by $\mS$ that has the ESSP and satisfies $\vert \mS\vert \leq \kappa$.
\end{lemma}
\begin{proof}
Let $Z=\{X_{\ell_0},\dots, X_{\ell_{\lambda-1}}\}$, where $\ell_0,\dots, \ell_{\lambda-1}\in \{0,\dots, n-1\}$, be a hitting set for $(\U,M)$, and let $B=(S',E,\delta',\iota)$ be the \LTS\  that originates from $A$ by removing, for all $j\in \{0,\dots, \lambda-1\}$, exactly the state $f_{\ell_j,1}$ (and the affected edges), and nothing else.
One easily verifies that $B$ is a state removal of $A$ by $\mS$ that satisfies $\vert \mS\vert \leq \kappa=\lambda$.
In particular, for every gadget of $A$, its initial state is present in $B$.
Hence, every region defined in for the proof of Lemma~\ref{lem:edge_removal_hs_implies_essp_and_ssp} is also valid for $B$.
Consequently, the claim follows similarly to the proof of this lemma.
\end{proof}

\subsection{State Removal aiming at embedding}\label{sec:state_removal_language_simulation}

Removing as few states of an \LTS\  as possible in order to make it implementable is also hard if we are aiming at embedding:
\begin{theorem}\label{the:state_removal_embedding}
{State Removal for Embedding} is NP-complete.
\end{theorem}

Similar to Theorem~\ref{the:state_removal_realization}, one argues that \textsc{State Removal for Language Simulation} is in NP.
We prove the hardness part by reusing the reduction of Section~\ref{sec:edge_removal_embedding}, that is, we define $\kappa=\lambda$, and let $A$ be defined as in Section~\ref{sec:edge_removal_embedding}.

\begin{lemma}\label{lem:state_removal_language_simulation_essp_implies_vc}
If there is a state removal $B$ for $A$ by $\mS$ that has the SSP and satisfies $\vert \mS\vert \leq \kappa$, then there is a hitting set with at most $\lambda=\kappa$ elements for $(\U,M)$.
\end{lemma}
\begin{proof}
Let $B$ be a state removal of $A$ that has the SSP and satisfies $\vert \mS\vert \leq \kappa$.
Again, without loss of generality, we may assume that $\mS\subseteq\{f_{j,1}|j\in\{0,\ldots,n-1\}\}$, so that in particular $B$ is reachable.
The proof is similar to the one of Lemma~\ref{lem:state_removal_essp_implies_hs}.

\medskip
In the following, we argue that
\[
Z=\{X_i\mid i\in \{0,\dots, n-1\} \text{ and }  f_{i,1}\in\mS\}
\]
defines a suitable hitting set for $(\U,M)$.

\medskip
Let $i\in \{0,\dots, m-1\}$ be arbitrary but fixed.
$T_{i,0}$, and $D_{i,0}$ are completely present in $B$.
Since $B$ has the SSP, there is a region $R=(sup, con, pro)$ that solves the SSA $(t_{i,0,0}, t_{i, 0, m_i+1})$.
Similar to the proof of Lemma~\ref{lem:edge_removal_ssp_implies_hs}, one argues that this implies that there is an $\ell\in \{1,\dots, m_i-1\}$, such that the edge $f_{i_\ell,0}\edge{X_{i_\ell}}f_{i_\ell,1}$ is not  in $B$, i.e., $X_{i_\ell}\in Z$.
By the arbitrariness of $i$, we get $Z\cap M_i\not=\emptyset$ for all $i\in \{0,\dots, m-1\}$, and $\vert Z\vert \leq\kappa=\lambda$.
\end{proof}

Conversely, we show:

\begin{lemma}\label{lem:state_removal_language_simulation_essp_implies_vc2}
If there is a hitting set with at most $\lambda$ elements for $(\U,M)$, then there is a state removal $B$ for $A$ by $\mS$ that has the SSP and satisfies $\vert \mS\vert \leq \kappa$,
\end{lemma}
\begin{proof}
Let $Z=\{X_{\ell_0}, \dots, X_{\ell_{\lambda-1}}\}$, where $\ell_0,\dots, \ell_{\lambda-1}\in \{0,\dots, n-1\}$ be a hitting set for $(\U,M)$, and let $B$ be the \LTS\  that originates from $A$ by the removal of the state $f_{\ell_j,1}$ (and the corresponding edges) for all $j\in \{0,\dots, \lambda-1\}$, and nothing else.
Obviously, $B$ is $A$ with at most $\kappa=\lambda$ state removals.
Moreover, for every gadget of $A$, its initial states is present in $B$.
Hence, the regions defined for the proof of Lemma~\ref{lem:edge_removal_hs_implies_ssp} are also valid here.
In particular, these regions witness the SSP of $B$.
The claim follows.
\end{proof}

\section{Parameterized complexity}\label{sec:para_complex}%

In the previous sections, we have shown that all the decision problems studied in this paper are not efficiently solvable (unless P=NP), i.e., there is no solution algorithm for these problems whose running time depends polynomially on the length of the input.
However, the input length alone is a fairly crude measure of the complexity of a problem and may make it appear more difficult than it actually~is.

In parameterized complexity, we examine the complexity of a problem not only in terms of the input length $n$, but also consider an additional (meaningful) parameter $\kappa$.
The question now is whether the problem is solvable by an algorithm whose ``non-polynomial behavior'' is functionally bounded by the parameter $\kappa$.
Such investigations are of great importance from both practical and theoretical points of view:
In the case of a positive decision, we usually obtain algorithms that can be considered efficient if $\kappa$ is small compared to the input length.
In the case of a negative decision, we know that we have to look for other solution strategies.
This prevents us from spending hours searching for algorithms that provably cannot exist.
In any case, from a theoretical standpoint, we get a better understanding of complexity at a more fine-grained level.

In the following, we introduce some necessary (basic) notions:

\begin{definition}\label{def:para_problem}
Let $\Sigma$ be an alphabet.
A \emph{parameterized problem} is a language $L\subseteq \Sigma^*\times \N$.
For an instance $(\omega, \kappa)\in \Sigma^*\times \N$, we call $\kappa$ the parameter.
\end{definition}

For all our decision problems, we consider $\kappa$ as the parameter, that is, $(A,\kappa)$ is an input of our original problems if and only if $((A,\kappa), \kappa)$ is an input of the parameterized versions.

\begin{definition}\label{def:fpt}
A parameterized problem $L\subseteq \Sigma^*\times \N$ is called \emph{fixed-parameter-tractable} (fpt for short),
if there is a constant $c$,
a polynomially computable function $f:\N\rightarrow \N$,
and an algorithm that correctly decides, for every $(\omega, \kappa)\in \Sigma^*\times \N$, in time $f(\kappa)\vert \omega\vert^c$, whether $(\omega,\kappa)\in L$.
The class of all fpt problems is called \emph{FPT}.
\end{definition}

Usually, before one tries to investigates whether a problem belongs to FPT, one proves that if $\kappa$ is fixed, then the problem is solvable in polynomial time, that is, it belongs the complexity class \emph{XP} (for \emph{slice-wise polynomial}).
It is easy to see, that if $\kappa$ is fixed, then all the synthesis up to removal problems we have considered here are solvable in polynomial time;
indeed, we only have to check for at most $\mathcal{O}((\vert S\vert+\vert E\vert)^\kappa)$ different modifications of $A$ whether they are implementable, and each check can be done in polynomial time in the size of $A$~\cite{tapsoft/BadouelBD95,txtcs/BadouelBD15}.

\medskip
In classical complexity, we prove, usually by means of a polynomial time reduction, that a problem is NP-hard to show that it (most likely) cannot be solved efficiently.
Similarly, we proceed in parameterized complexity theory to argue that a parameterized problem is not fpt.
For every $i\in \mathbb{N}^+$, $\w[i]$ defines a class of parameterized problems and the following inclusions hold:
\[
\text{FPT}\subseteq \w[1]\subseteq \w[2]\subseteq \w[3]\dots \subseteq \text{XP}
\]

The precise definition of the class $\w[i]$, $i\in \mathbb{N}^+$, is beyond the scope of this paper.
We therefore refer the reader to~\cite{sp/CyganFKLMPPS15}.
It can be proved that $\text{FPT}\not=\text{XP}$.
Moreover, according to the generally accepted working hypothesis, even $\text{FPT}\not=\w[1]$ and $\w[i]\not=\w[i+1]$ holds for all $i\in \mathbb{N}^+$.
Consequently, for all $i\geq 1$, we assume that any $\w[i]$-hard problem is not fpt.
Thereby, we call a parameterized problem $L_2$ $\w[i]$-hard if every parameterized problem $L_1$ from $\w[i]$ can be reduced to $L_2$ by means of a \emph{parameterized reduction}:

\begin{definition}\label{def:para_reduction}
Let $L_1,L_2\subseteq \Sigma^*\times\mathbb{N}$ be two parameterized problems.
We say $L_1$ can be \emph{reduced} to $L_2$ by a \emph{parameterized reduction}, if there are a constant $c$,
polynomially computable functions   $g,h:\mathbb{N}\rightarrow \mathbb{N}$ and a function $f:\Sigma^*\times\mathbb{N}\rightarrow \Sigma^*\times\mathbb{N}$, such that, for all $(\omega,\kappa)\in  \Sigma^*\times\mathbb{N}$, the following conditions are true:
\begin{enumerate}
\itemsep=0.95pt
\item $(\omega,\kappa)\in L_1$ if and only if $f(\omega,\kappa)=(\omega',\kappa')\in L_2$;
\item
$f$ is computable in time  $g(\kappa)\cdot \vert \omega\vert^c$;
\item
$\kappa'\leq h(\kappa)$.
\end{enumerate}
\end{definition}

\noindent
Moreover, if the problem also belongs to $W[i]$, then we say it is $\w[i]$-complete.
To show that a parameterized problem $L_2$ is not fpt, a classical technique is to show for a known $W[i]$-hard problem $L_1$, where $i\in \mathbb{N}^+$, that it can be reduced to $L_2$ by a parameterized reduction.
Because $L_1$ is $W[i]$-hard and hence (most likely) not fpt, it holds according to the following lemma that $L_2$ is also not fpt:
\begin{lemma}[\cite{sp/CyganFKLMPPS15}]\label{lem:para_reduction}
If $L_1$ and $L_2$ are two parameterized problems such that $L_1$ is reducible to $L_2$ by a parameterized reduction and $L_2\in \text{FPT}$, then $L_1\in \text{FPT}$.
\end{lemma}

By the following theorem, the problem \textsc{Hitting Set} is most likely not fpt:

\begin{theorem}[\cite{sp/CyganFKLMPPS15}]\label{the:hs_para}
{Hitting Set} parameterized by $\lambda$ is $W[2]$-complete.
\end{theorem}

Obviously, the reductions from the previous sections are parameterized reductions, since $\kappa=\lambda$.
Hence, by Lemma~\ref{lem:para_reduction}, and Theorem~\ref{the:hs_para}, we obtain the following result:

\begin{theorem}\label{the:para}
For every of the  edge-,  event- and state removal modifications,
and for every of the implementations embedding, language-simulation and realization,
the parameterized problem of deciding, for a given LTS $A$, and a natural number $\kappa$, whether $A$ can be made implementable by removing at most $\kappa$ components belongs to the complexity class $W[P]$, and is $W[2]$-hard, where $\kappa$ is considered as the parameter.
\end{theorem}

\section{Inapproximability}\label{sec:inapproximability}%

For each modification technique \emph{edge removal}, \emph{event removal} and \emph{state removal}, and for each implementation \emph{embedding}, \emph{language-simulation} and \emph{realization}, our goal is to modify $A$ as little as possible according to the chosen modification kind, so that the result is an implementable \LTS.
Since a lower bound of this optimization problem is defined by the complexity of its decision version, there is little hope that an implementable modification can be found efficiently according to the results of the previous sections.
Moreover, by Theorem~\ref{the:para}, they are most likely not fpt for the natural parameter~$\kappa$.

On the other hand, instead of an exact solution, an approximate solution could be satisfactory,
at least if a fixed approximation ratio $c$ can be guaranteed.
This leads to the search for a so-called $c$-approximation algorithm, i.e., an algorithm that, for a given \LTS\  $A$, yields in polynomial time an implementable  modification $B$ of $A$ with at most $c\cdot\kappa$ removed components  when the optimal solution is $\kappa$, where $c\geq 1$ is a fixed constant.
In this section we will argue that no such algorithm exists for any of our problems.

For a start, we provide some necessary basic definitions.

\begin{definition}[Optimization Problem~\cite{siamcomp/CrescenziKST99}]\label{def:npo_problem}
An \emph{NP optimization problem} $Q$ is a 4-tuple $(I,sol,m,type)$ such that the following hold:
\begin{enumerate}
\itemsep=0.95pt
\item
$I$ is the set of \emph{instances} of $Q$, and it is recognizable in polynomial time.
\item
Given an instance $x\in I$, $sol(x)$ denotes the set of \emph{feasible solutions} of $x$.
The solutions of $x$ are short, that is, there exists a polynomial $p$ such that, for every $y\in sol(x)$, it holds $\vert y\vert \leq p(\vert x\vert)$ and, 
for every $x$ and every $y$ with $\vert y\vert \leq p(\vert x\vert)$, it is decidable in polynomial time whether $y\in sol(x)$.
\item
For $x\in I$, and $y\in sol(x)$, $m(x,y)$ denotes the positive integer \emph{measure} of $y$.
The function $m$ is computable in polynomial time.
\item
$type\in \{min,max\}$.
\end{enumerate}
\end{definition}

The \emph{class} NPO is the set of all NP optimization problems, and \minnpo\ is the set of all minimization NPO problems.
In this paper, we deal only with problems of \minnpo.
The goal of a \minnpo\ problem $Q=(I,sol, m, min)$ with respect to an instance $x$ is to find an \emph{minimum solution}, that is, an $y^*\in sol(x)$, such that $m(x,y^*)={min}\{m(x,y)\mid y\in sol(x)\}$.
In the following, $\opt$ will denote the function that maps an instance $x$ of $Q$ to the measure of a minimum solution.

\begin{definition}[Ratio]\label{def:ratio}
Let $Q=(I,sol, m, min)$ be an Min-NPO problem.
For every $x\in I$ and $y\in sol(x)$, by  \[R(x,y)=\frac{m(x,y)}{\opt(x)}\] we define the \emph{performance ratio} $R$ of $y$ with respect to $x$.
\end{definition}

\begin{definition}[c-Approximation Algorithm]
Given an \minnpo\ problem $Q$, and an arbitrary constant $c\geq 1$, we say that an algorithm $T$ is a $c$-approximation algorithm for $Q$, if it runs in polynomial time and, for every instance $x$ of $Q$ with $sol(x)\not=\emptyset$, outputs a solution $T(x)\in sol(x)$ such that $R(x,T(x))\leq c$.
\end{definition}

The following notion of \emph{L-reducibility} can be used to show that, for any constant $c\geq 1$, a \minnpo\ problem does not have a $c$-approximation algorithm:

\begin{definition}[L-reduction~\cite{siamcomp/CrescenziKST99} ]\label{def:l_reduction}
Let $Q_1=(I_{Q_1}, sol_{Q_1}, m_{Q_1}, min)$ and $Q_2=(I_{Q_2}, sol_{Q_2}, m_{Q_2}, min)$ be two \minnpo\ problems.
$Q_1$ is said to be \emph{L-reducible} to $Q_2$, in symbols $Q_1\leq Q_2$, if there exist two functions $g,f$ and two positive constants $\alpha, \beta$ such that:
\begin{enumerate}
\itemsep=0.9pt
\item
For every $x\in I_{Q_1}$, $f(x)\in I_{Q_2}$ is computable in polynomial time.
\item
For every $x\in I_{Q_1}$ and every $y\in sol_{Q_2}(f(x))$, $\;g(x,y)\in sol_{Q_1}(x)$ is computable in polynomial time.
\item
For every $x\in I_{Q_1}$, $\opt_{Q_2}(f(x))\leq \alpha\cdot \opt_{Q_1}(x)$
\item
For every $x\in I_{Q_1}$ and every $y\in sol_{Q_2}(f(x))$,
$$
\vert \opt_{Q_1}(x)-m_{Q_1}(x, g(x,y))\vert \leq \beta\cdot \vert \opt_{Q_2}(f(x))-m_{Q_2}(f(x),y)\vert
$$
\end{enumerate}
The 4-tuble $(f,g,\alpha,\beta)$ is called an \emph{L-reduction} from $Q_1$ to $Q_2$.
\end{definition}

It is known from~\cite[p.~1765]{siamcomp/CrescenziKST99} that if $Q_1\leq_L Q_2$ and there exists a $c$-approximation algorithm for $Q_2$, where $c\geq 1$ is a constant, then there exists a $c'$-algorithm for $Q_1$, where $c'=1+\alpha\cdot\beta\cdot(c-1)$.

By Corollary~1.5 of~\cite{corr/abs-1305-1979},  for every $c\geq 1$, there does not exist a $c$-approximation algorithm of the problem \textsc{Minimum Set Cover}, unless $\text{P}=\text{NP}$:

\noindent
\fbox{\begin{minipage}[t][1.8\height][c]{0.97\textwidth}
\begin{optproblem}
  \oproblemtitle{\textsc{Minimum Set Cover}}
  \oprobleminput{A finite set $\U$, and a system $M$ of subsets of $\U$.}
  \oproblemsolution{$Y\subseteq M$ such that  $\bigcup_{y\in Y}y=\U$.}
  \oproblemmeasure{$\vert Y\vert$}
\end{optproblem}
\end{minipage}}\smallskip

\noindent
Hence, by the well-known equivalence of \textsc{Minimum Set Cover} and \textsc{Minimum Hitting Set}, the same is true for the latter problem:

\noindent
\fbox{\begin{minipage}[t][1.8\height][c]{0.97\textwidth}
\begin{optproblem}
  \oproblemtitle{\textsc{Minimum Hitting Set}}
  \oprobleminput{A finite set $\U$, and a system $M$ of subsets of $\U$.}
  \oproblemsolution{A hitting set $Z\subset \U$ for $(\U,M)$.}
  \oproblemmeasure{$\vert Z\vert$}
\end{optproblem}
\end{minipage}}\smallskip

In the following, we argue that, for every component \textsc{X} of \textsc{Edge, Event} or \textsc{State} and for every implementation \textsc{Y} of \textsc{Embedding, Language Simulation} or \textsc{ Realization}, there is an L-reduction from \textsc{Minimum Hitting Set} to the following \minnpo\ problem \textsc{X-Removal for Y}, and thus prove that the latter does not allow a $c$-approximation algorithm for any constant $c\geq 1$:

\noindent
\fbox{\begin{minipage}[t][1.8\height][c]{0.97\textwidth}
\begin{optproblem}
  \oproblemtitle{\textsc{Minimum  X-Removal for Y}}
  \oprobleminput{An \LTS\  $A=(S,E,\delta,\iota)$.}
  \oproblemsolution{A \textsc{Y}-implementable \textsc{X-Removal} $B$ of $A$.}
  \oproblemmeasure{Number of components removed (from $A$ to obtain $B$ according  to \textsc{X}).}
\end{optproblem}
\end{minipage}}\smallskip

Notice that the measure is actually a function that depends on $A$ and $B$ since, for example, $m(A,B)=\vert E\setminus E'\vert $ if $B$ is an event removal of $A$.

\smallskip
The proof of the following lemma is based on the observation that the polynomial time reductions of the previous sections can be extended to suitable L-reductions:

\begin{lemma}\label{lem:l_reduction}
For every component \textsc{X} of $\{\textsc{Edge, Event, State}\}$ and for every implementation \textsc{Y} of $\{\textsc{Embedding, Language Simulation, Realization}\}$, there is an L-reduction from \textsc{Minimum Hitting Set} to \textsc{X-Removal for Y}.
\end{lemma}
\begin{proof}
For a start, we argue for the problem \textsc{Minimum Edge Removal for Realization}; the argument will be similar for the other cases.

We obtain $(f,g,\alpha,\beta)$ as follows:
$f$ is the function that maps an input $(\U,M)$ to the \LTS\ $A$ in accordance to the reduction of Section~\ref{sec:edge_removal_realization};
for every solution $B$ of $A$, $g(A,B)=\{X\in \U\mid \exists s\edge{X}s' \in \K\}$; finally, $\alpha=\beta=1$.
Obviously, $f$ and $g$ can be computed in polynomial time and, by the proof arguments of Lemma~\ref{lem:edge_removal_essp_implies_hs}, $g(A,B)$ defines actually a hitting set of $(\U,M)$.
We argue that the optimum values of both problems are equal:

Let $Z$ be an optimum hitting set for $(\U,M)$, that is, $\vert Z\vert =\opt((\U,M))$, and let $B=(S,E',\delta',\iota)$ be the edge removal of $A$ obtained in accordance to the proof arguments of Lemma~\ref{lem:edge_removal_hs_implies_essp_and_ssp} with the set of removed edges $\K$.
By the arguments of the same lemma, $B$ is a feasible solution for $A$, which satisfies $\vert \K\vert =\vert Z\vert$.
Assume that there is a realizable solution $B'=(S, E'',\delta'',\iota)$ of $A$ with set of removed edges $\K'$ such that $\vert \K'\vert < \vert \K\vert$, i.e., $B'$ removes strictly less edges than $B$.
By the proof of Lemma~\ref{lem:edge_removal_essp_implies_hs}, the set $Z'=\{X\in \U\mid s\edge{X}s'\in \K'\}$ defines a hitting set for $(\U,M)$.
This implies $\vert Z'\vert\leq \vert \K'\vert < \vert \K\vert =\vert Z\vert$ and thus contradicts the choice of $Z$.
Hence, $\opt((\U,M))=\opt(A)$, which obviously implies $\opt(A)\leq \alpha\cdot \opt((\U,M))$.
In particular, $(f,g,\alpha,\beta)$ satisfies Condition~3 of Definition~\ref{def:l_reduction}.

\medskip
Moreover, if $B'=(S, E'',\delta'',\iota)$ is an arbitrary but fixed implementable solution of $A$ with a set of removed edges $\K'$, then $g(A,B')=Z'=\{X\in \U\mid s\edge{X}s'\in \K'\}$ and we get:
\[
\vert \opt((\U,M))-\vert Z'\vert\vert = \vert \opt(A)-\vert Z'\vert \vert \leq \vert \opt(A)-\vert \K'\vert \vert = \beta\cdot \vert \opt(A)-\vert \K'\vert \vert \]
since $0 \leq \opt((\U,M))=\opt(A)\leq \vert Z'\vert \leq \vert \K'\vert$.
Hence, $(f,g,\alpha,\beta)$ satisfies Condition~4 of Definition~\ref{def:l_reduction}, and thus is a valid L-reduction.\\

The arguments for the other cases are similar:
For every remaining combination of component X and implementation Y, one finds out that there is an L-reduction $(f,g,1,1)$, such that
\begin{itemize}
\itemsep=0.9pt
\item
$f$ is defined in accordance to the corresponding reduction that maps $(\U,M)$ to \LTS\ $A$, and
\item
for any Y-implementable X-removal $B$ of $A$, $g$ maps $A$, and $B$ to the set $Z=g(A,B)\subseteq \U$ of elements that is implied by the set $\R$ of removed components, and
\item
$\opt((\U,M))=\opt(A)$, and $\vert \opt((\U,M))-\vert Z\vert\vert \leq \vert \opt(A)-\vert \R\vert \vert$.
\end{itemize}

\vspace*{-7mm}
\end{proof}

Altogether, we get the main result of this section:

\begin{theorem}\label{the:inapproximability}
For every component \textsc{X} of \textsc{Edge, Event, State}, for every implementation \textsc{Y} of \textsc{Embedding, Language Simulation, Realization} and for every constant $c\geq 1$, the problem \textsc{X-Removal for Y} does not allow a $c$-approximation algorithm, unless $\text{P}=\text{NP}$.
\end{theorem}
%

\section{Conclusion}\label{sec:conclusion}%

In this paper, we showed that converting an unimplementable \LTS\ into an implementable one by removing
as few of its states, events or edges as possible, is intractable from the point of view of classical complexity theory.
This solves a problem that was left open in~\cite{corr/abs-2002-04841}.
Notice that the reductions for edge- and event removal also work if these modifications are defined in a way that require all original states to be preserved.
However, in general, they could then not always produce an implementable LTS, since there are unimplementable trees.
Moreover, we also show that these problems are also intractable from the point of view of  parameterized complexity, and approximability as well.
However, a complete characterization of the parameterized complexity of the problems is still open, and future work should address finding the exact upper bounds for them.
Future work could also address other techniques of modifications that were suggested in the literature such as, for example, state- or event-refinement, state fusion or edge addition (like in~\cite{dnb/Schlachter18}), or a mixture of these modifications.

\subsubsection*{Acknowledgements.}%

We would like to thank the anonymous reviewers of the original version of this paper, and this extended version for their detailed comments and valuable suggestions.



\end{document}